\documentclass[fleqn,10pt]{llncs}

\usepackage{amsmath}
\usepackage{amssymb}
\usepackage{cite}
\usepackage{enumerate}
\usepackage{fancyvrb}
\usepackage[T1]{fontenc}
\usepackage[utf8]{inputenc}
\usepackage{setspace}
\usepackage{tikz}
\usetikzlibrary{trees,arrows,arrows.meta,shapes,decorations.pathmorphing,backgrounds,positioning,fit,petri}
\usepackage{todonotes}
\usepackage[normalem]{ulem}   
\usepackage{url}
\usepackage[all]{xy}
\usepackage{mathtools}
\usepackage{bm}
\usepackage{caption}
\usepackage{subcaption}
\usepackage{relsize}
\usepackage{stmaryrd}
\usepackage{mathpartir}
\usepackage{graphicx}
\usepackage{float}
\usepackage{hyperref}
\usepackage{listings}
\usepackage{syntax}
\usepackage{mathabx}

\graphicspath{{images/}}
\tikzset{%
	>={Latex[width=2mm,length=2mm]},
	module/.style = {rectangle, draw, minimum height=0.8cm, 
		minimum width=2.5cm, fill=orange!15, text centered, 
		font=\ttfamily},
	store/.style = {circle, draw, minimum height=1cm, 
		fill=orange!20, text centered, font=\ttfamily},
	curr/.style = {circle, draw, minimum height=1cm, 
		fill=red!30, text centered, font=\ttfamily},
	triplearrow/.style={
		draw=black!75,
		color=black!75,
		double distance=3pt,
		postaction={draw=black!75, color=black!75}, 
		->},
}

\usepackage{proof} 
\usepackage{myproof}

\newcommand{\CS}{\textnormal{CS}}
\newcommand{\SCS}{\textnormal{SCS}}
\newcommand{\SCSE}{\textnormal{SCSE}}
\newcommand{\CCP}{\textnormal{CCP}}
\newcommand{\SCCP}{\textnormal{SCCP}}
\newcommand{\SCCPwE}{\textnormal{SCCP\texttt{+}E}}

\newcommand{\EQ}{=}
\newcommand{\IFS}{\mathbf{if}}
\newcommand{\REW}{\rightarrow}
\newcommand{\eq}[2]{#1 \EQ #2}

\newcommand{\crl}[3]{#1 \REW #2 \; \IFS\; #3}

\newcommand{\cond}{\gamma}

\newcommand{\ccrl}[3]{#1 \REW #2\;\textnormal{\bf if}\; {#3}}

\newcommand{\can}[2]{{#1}\!\downarrow_{#2}}

\newcommand{\ecal}{\mathcal{E}}
\newcommand{\lcal}{\mathcal{L}}
\newcommand{\rcal}{\mathcal{R}}
\newcommand{\tcal}{\mathcal{T}}
\newcommand{\ls}[1]{\mathit{ls}(#1)}
\newcommand{\ded}{\vdash}
\newcommand{\sort}[1]{\textit{#1}}
\newcommand{\states}{\sort{State}}
\newcommand{\rews}{\rightarrow}
\newcommand{\func}[3]{#1 : #2 \longrightarrow #3}
\newcommand{\lang}{{0}}
\newcommand{\nlang}{{1}}
\newcommand{\oqff}[2]{\textit{\textit{QF}}_{#1}({#2})}
\newcommand{\MODELS}{\vDash}
\newcommand{\RMODELS}{\Dashv}

\DeclarePairedDelimiter{\sfunc}{\textbf{\textup{[}}}{\textbf{\textup{]}}}					
\DeclarePairedDelimiter{\efunc}{\uparrow}{}					

\newcommand{\true}{\it{true}}								
\newcommand{\false}{\it{false}}								
\newcommand{\Con}{\it{Con}}									
\newcommand{\join}[2]{#1 \ \sqcup \ #2}						
\newcommand{\cleq}[2]{#1 \sqsubseteq #2}					
\newcommand{\cgeq}[2]{#1 \sqsupseteq #2}					

\newcommand{\sfuncs}{\sfunc{\cdot}_1,\ldots,\sfunc{\cdot}_n}	
\newcommand{\efuncs}{\efunc{_1},\ldots,\efunc{_n}}				
\newcommand{\defsymbol}{\stackrel{\textup{\texttt{def}}}  {=}}	
\newcommand{\entails}{\vdash}

\newcommand{\conf}[2]{\langle #1; #2\rangle}
\newcommand{\redi}{\longrightarrow }

\newcommand{\extr}[2]{#2\mathop{\uparrow_#1}}

\newcommand{\rTell}{\sf {Tell}}
\newcommand{\rAsk}{\sf {Ask}}
\newcommand{\rPar}{\sf {Par}}
\newcommand{\rSp}{\sf {SP}}
\newcommand{\rRec}{\sf {Rec}}
\newcommand{\rExt}{\sf {Ext}}

\newcommand{\mth}[1]{\textnormal{\textit{#1}}}
\newcommand{\cde}[1]{\textnormal{\texttt{#1}}}
\newcommand{\rlname}[1]{\cde{[#1]}}

\lstdefinelanguage{Maude}{%
   keywords={
    , mod, fmod, endm, endfm
    , pr , protecting 
    , ex , extending 
    , inc, including
    , sort, sorts, subsort, subsorts
    , var, vars
    , op, ops
    , eq, ceq
    , rl, crl
    , if
    , search
    , red, reduce
    }
}
\lstnewenvironment{maude}
{\lstset{language=Maude ,
    keywordstyle=\color{blue}
  , basicstyle=\ttfamily\singlespacing\small
  , commentstyle={}
  , columns=flexible
  , numbers=none
  , showstringspaces=false
  , keepspaces=true
  , aboveskip=-3pt
  } 
}
{}

\DefineVerbatimEnvironment{maude2}{Verbatim}{fontsize=\scriptsize}

\lstdefinelanguage{Sccp}{%
}

\lstnewenvironment{sccp}
{\lstset{language=SCCP
		, keywords={begin, end, var}	
		, keywordstyle=\color{red}
		, keywords=[2]{ask, tell, Int, Bool, x, v, r}
		, keywordstyle=[2]\color{blue}
		, basicstyle=\ttfamily\singlespacing\small
		, commentstyle={}
		, columns=flexible
		, numbers=none
		, showstringspaces=false
		, keepspaces=true
		, aboveskip=-3pt
		, breaklines=true
	} 
}
{}

\DefineVerbatimEnvironment{sccp2}{Verbatim}{fontsize=\scriptsize}

\newcommand{\tellp}[1]{\tell(#1)}
\newcommand{\askp}[2]{\ask \  #1 \  \rightarrow \ #2}

\renewcommand{\iff}{\mbox{\ \ iff \ \ }}

\newcommand{\C}{\mathcal{C}}

\newcommand{\ask}{{\bf ask}}

\newcommand{\tell}{{\bf tell}}

\newcommand{\Stop}{{\bf 0}}

\newcommand{\rrarrow}{\longrightarrow}

\newcommand{\pairccp}[2]{\langle #1,#2 \rangle}

\long\def\comment#1{}

\newcommand{\To}{\Rightarrow}

\newcommand{\K}[2]{\left[ #2 \right] _{#1}}			

\newcommand{\prj}[2]{#1^{{#2}}}

\makeatletter

\newdimen\w@dth

\def\setw@dth#1#2{\setbox\z@\hbox{\scriptsize $#1$}\w@dth=\wd\z@
\setbox\@ne\hbox{\scriptsize $#2$}\ifnum\w@dth<\wd\@ne \w@dth=\wd\@ne \fi
\advance\w@dth by 1.2em}

\def\t@^#1_#2{\allowbreak\def\n@one{#1}\def\n@two{#2}\mathrel
{\setw@dth{#1}{#2}
\mathop{\hbox to \w@dth{\rightarrowfill}}\limits
\ifx\n@one\empty\else ^{\box\z@}\fi
\ifx\n@two\empty\else _{\box\@ne}\fi}}
\def\t@@^#1{\@ifnextchar_ {\t@^{#1}}{\t@^{#1}_{}}}

\def\t@left^#1_#2{\def\n@one{#1}\def\n@two{#2}\mathrel{\setw@dth{#1}{#2}
\mathop{\hbox to \w@dth{\leftarrowfill}}\limits
\ifx\n@one\empty\else ^{\box\z@}\fi
\ifx\n@two\empty\else _{\box\@ne}\fi}}
\def\t@@left^#1{\@ifnextchar_ {\t@left^{#1}}{\t@left^{#1}_{}}}

\def\two@^#1_#2{\def\n@one{#1}\def\n@two{#2}\mathrel{\setw@dth{#1}{#2}
\mathop{\vcenter{\hbox to \w@dth{\rightarrowfill}\kern-1.7ex
                 \hbox to \w@dth{\rightarrowfill}}%
       }\limits
\ifx\n@one\empty\else ^{\box\z@}\fi
\ifx\n@two\empty\else _{\box\@ne}\fi}}
\def\tw@@^#1{\@ifnextchar_ {\two@^{#1}}{\two@^{#1}_{}}}

\def\tofr@^#1_#2{\def\n@one{#1}\def\n@two{#2}\mathrel{\setw@dth{#1}{#2}
\mathop{\vcenter{\hbox to \w@dth{\rightarrowfill}\kern-1.7ex
                 \hbox to \w@dth{\leftarrowfill}}%
       }\limits
\ifx\n@one\empty\else ^{\box\z@}\fi
\ifx\n@two\empty\else _{\box\@ne}\fi}}
\def\t@fr@^#1{\@ifnextchar_ {\tofr@^{#1}}{\tofr@^{#1}_{}}}

\newdimen\W@dth
\def\setW@dth#1#2{\setbox\z@\hbox{$#1$}\W@dth=\wd\z@
\setbox\@ne\hbox{$#2$}\ifnum\W@dth<\wd\@ne \W@dth=\wd\@ne \fi
\advance\W@dth by 1.2em}

\def\T@^#1_#2{\allowbreak\def\N@one{#1}\def\N@two{#2}\mathrel
{\setW@dth{#1}{#2}
\mathop{\hbox to \W@dth{\rightarrowfill}}\limits
\ifx\N@one\empty\else ^{\box\z@}\fi
\ifx\N@two\empty\else _{\box\@ne}\fi}}
\def\T@@^#1{\@ifnextchar_ {\T@^{#1}}{\T@^{#1}_{}}}

\def\T@left^#1_#2{\def\N@one{#1}\def\N@two{#2}\mathrel{\setW@dth{#1}{#2}
\mathop{\hbox to \W@dth{\leftarrowfill}}\limits
\ifx\N@one\empty\else ^{\box\z@}\fi
\ifx\N@two\empty\else _{\box\@ne}\fi}}
\def\T@@left^#1{\@ifnextchar_ {\T@left^{#1}}{\T@left^{#1}_{}}}

\def\Tofr@^#1_#2{\def\N@one{#1}\def\N@two{#2}\mathrel{\setW@dth{#1}{#2}
\mathop{\vcenter{\hbox to \W@dth{\rightarrowfill}\kern-1.7ex
                 \hbox to \W@dth{\leftarrowfill}}%
       }\limits
\ifx\N@one\empty\else ^{\box\z@}\fi
\ifx\N@two\empty\else _{\box\@ne}\fi}}
\def\T@fr@^#1{\@ifnextchar_ {\Tofr@^{#1}}{\Tofr@^{#1}_{}}}

\def\Two@^#1_#2{\def\N@one{#1}\def\N@two{#2}\mathrel{\setW@dth{#1}{#2}
\mathop{\vcenter{\hbox to \W@dth{\rightarrowfill}\kern-1.7ex
                 \hbox to \W@dth{\rightarrowfill}}%
       }\limits
\ifx\N@one\empty\else ^{\box\z@}\fi
\ifx\N@two\empty\else _{\box\@ne}\fi}}
\def\Tw@@^#1{\@ifnextchar_ {\Two@^{#1}}{\Two@^{#1}_{}}}

\def\to{\@ifnextchar^ {\t@@}{\t@@^{}}}
\def\from{\@ifnextchar^ {\t@@left}{\t@@left^{}}}
\def\two{\@ifnextchar^ {\tw@@}{\tw@@^{}}}
\def\tofro{\@ifnextchar^ {\t@fr@}{\t@fr@^{}}}
\def\To{\@ifnextchar^ {\T@@}{\T@@^{}}}
\def\From{\@ifnextchar^ {\T@@left}{\T@@left^{}}}
\def\Two{\@ifnextchar^ {\Tw@@}{\Tw@@^{}}}
\def\Tofro{\@ifnextchar^ {\T@fr@}{\T@fr@^{}}}

\makeatother

\title{ Reachability Analysis for Spatial Concurrent Constraint
  Systems with Extrusion }

\author{
  Miguel Romero \and
  Camilo Rocha
}

\institute{
  Department of Electronics and Computer Science \\
  Pontificia Universidad Javeriana \\
  Cali, Colombia \\
  \email{\{miguel.romero,camilo.rocha\}@javerianacali.edu.co}
}

\begin{document}

\maketitle

\begin{abstract}
  Spatial concurrent constraint programming ($\SCCP$) is an algebraic
  model of spatial modalities in constrained-based process calculi; it
  can be used to reason about spatial information distributed among
  the agents of a system. This work presents an executable rewriting
  logic semantics of $\SCCP$ with extrusion (i.e., process mobility)
  that uses rewriting modulo SMT, a novel technique that combines the
  power of term rewriting, matching algorithms, and SMT-solving. In
  this setting, constraints are encoded as formulas in a theory with a
  satisfaction relation decided by an SMT solver, while the topology
  of the spatial hierarchy is encoded as part of the term structure of
  symbolic states. By being executable, the rewriting logic
  specification offers support for the inherent symbolic and
  challenging task of reachability analysis in the constrained-based
  model. The approach is illustrated with examples about the automatic
  verification of fault-tolerance, consistency, and privacy in
  distributed spatial and hierarchical systems.
\end{abstract}
\section{Introduction}
\label{sec.intro}

The widespread availability of virtualization resources such as
container and virtual machine technology are marking a new incarnation
of distributed systems. Tasks such as fault-tolerant infrastructure
monitoring and delivery of goods in unmanned aerial systems,
preventing privacy breaches in social networks and cloud storage by
managing information access, and saving lives by controlling and
monitoring pace makers are now taking place in the presence of spatial
hierarchies. This means that the usual high degree of safety criteria
in such systems is now exposed to the presence of hierarchical
computation (and sharing) of information among groups of distributed
and concurrent agents, which makes it an even more challenging goal
for formal modeling and verification purposes.

An interesting step towards mathematically understanding and formally
modeling these highly distributed hierarchical systems has been taken
by S. Knight et al.~\cite{knight-sccp-2012} and M. Guzmán et
al.~\cite{guzman-sccpe-2017}. They introduce an algebraic model of
spatial modalities in constrained-based process calculi where
information (e.g., knowledge) can be shared in spatially distributed
agents that interact with the global system by launching processes
(e.g., programs). The scope of an agent is given by a spatial operator
indicating where a process resides within the space structure, where
it queries and posts information in the local store. In the end, their
proposal offers an algebraic framework to model and reason about
important concepts of safety-critical systems such as fault-tolerance,
consistency, and privacy within the setting of distributed
hierarchical spaces and process extrusion (i.e., mobility).

This work addresses the key issue of automatically verifying
reachability properties of distributed hierarchical systems based on
the algebraic model of spatially constrained concurrent process with
extrusion in~\cite{knight-sccp-2012,guzman-sccpe-2017}. The approach
is based on the formal specification of such a model as a theory in
rewriting logic~\cite{meseguer-rltcs-1992}, a semantic framework
unifying a wide range of models of concurrency. The formal
specification is executable in Maude~\cite{clavel-maudebook-2007},
thus benefiting from formal analysis techniques and tools such as
state-space exploration and automata-based LTL model checking. Safety
criteria such as fault-tolerance (e.g., when and how does a local
store first become inconsistent?), consistency (e.g., does a fault
propagate to the global system?, do two stores have the same
information?), and privacy (e.g., does a store ever gain enough
information as to reveal private information?) can now be
automatically queried in these systems.

The rewriting logic theory introduced in this work supports the
constructs of constrained-based process calculi presented in
in~\cite{knight-sccp-2012,guzman-sccpe-2017} such as posting and
querying information from/to a local store, parallel composition of
processes, recursion, and extrusion. The underlying constraint system
is materialized with the help of SMT-solving technology. In
particular, the constraints are quantifier-free formulas interpreted
over the Booleans and integers, and the information entailment in the
algebraic model is realized via semantic inference. By following the
rewriting modulo SMT approach~\cite{rocha-rewsmtjlamp-2017},
simulation of the rewrite relation induced by the rewriting logic
theory can be performed efficiently using matching and SMT-solving.
In this setting, existential reachability queries can be automatically
and efficiently performed using existing rewrite-based facilities
available from Maude.

Continuing with the effort of putting epistemic concepts in the hands
of programmers initiated in~\cite{knight-sccp-2012}, a programming
language is introduced in this work. This language provides
programmers with the building blocks of a language based on the
algebraic model of spatial modalities in constrained-based process
calculi, with the executable semantics given by rewrite theory
above-mentioned. Such a language is accompanied with an open-source
graphical environment in which code can be edited and executed.

\paragraph{Outline.}
This work is organized as follows. Section~\ref{sec.prelim} presents
some preliminaries on concurrent constraint programming, rewriting
logic, and SMT-solving. Section~\ref{sec.sccp} overviews spatial
concurrent constraint systems with extrusion. Section~\ref{sec.rew}
introduces the rewriting logic semantics of the algebraic model and
Section~\ref{sec.rew} presents the mechanical proofs obtained for
ensuring the correctness of the operational semantics of the rewrite
theory w.r.t. to its mathematical one. Section~\ref{sec.reach}
explains how existential reachability properties can be automatically
proved and examples illustrating such a feature. The language and tool
based on the algebraic model are explained in Section~\ref{sec.lang}.
Finally, Section~\ref{sec.concl} concludes the work.

\section{Preliminaries}
\label{sec.prelim}

\subsection{Concurrent Constraint Programming and Constraint Systems}

{\it Concurrent Constraint Programming}
($\CCP$)~\cite{saraswat-ccp-1990,saraswat-ccpbook-1993,saraswat-ccpsem-1991}
(see a survey in \cite{olarte-emergmodels-2013}) is a model for
concurrency that combines the traditional operational view of process
calculi with a {\it declarative} view based on logic. This allows
$\CCP$ benefit from the large set of reasoning techniques of both
process calculi and logic. Under this paradigm, the conception of
\emph{store as valuation} in the von Neumann model is replaced by the
notion of \emph{store as constraint} and processes are seen as
information transducers.

The $\CCP$ model of computation makes use of \emph{ask} and
\emph{tell} operations instead of the classical read and write. An ask
operation tests if a given piece of information (i.e., a constraint as
in $temperature > 23$) can be deduced from the store.  The tell
operations post constraints in the store, thus augmenting/refining the
information in it.  A fundamental issue in $\CCP$ is then the
specification of systems by means of constraints that represent
partial information about certain variables. The state of the system
is specified by the store (i.e., a constraint) that is monotonically
refined by processes adding new information.
  
The basic constructs (processes) in \emph{$\CCP$} are: (1) the ${\it
  tell}(c)$ agent, which posts the constraint $c$ to the store, making
it available to the other processes. Once a constraint is added, it
cannot be removed from the store (i.e., the store grows
monotonically). And (2), the ask process $c\rightarrow P$, which
queries if $c$ can be deduced from the information in the current
store; if so, the agent behaves like $P$, otherwise, it remains
blocked until more information is added to the store. In this way, ask
processes define a reactive synchronization mechanism based on
entailment of constraints. A basic $\CCP$ process language usually
adds \emph{parallel composition} ($P\parallel Q$) combining processes
concurrently, a \emph{hiding} operator for local variable definition,
and potential infinite computation by means of recursion or
replication.

The $\CCP$ model is parametric in a \emph{constraint system} ($\CS$)
specifying the structure and interdependencies of the partial
information that processes can query (\emph{ask}) and post
(\emph{tell}) in the \emph{shared store}. The notion of constraint
system can be given by using first-order logic. Given a signature
$\Sigma$ and a first-order theory $\Delta$ over $\Sigma$, constraints
can be thought of as first-order formulae over $\Sigma$. The (binary)
entailment relation $\vdash$ over constraints is defined for any pair
of constraints $c$ and $d$ by $c\entails d$ iff the implication $c
\Rightarrow d$ is valid in $\Delta$. As an example, take the finite
domain constraint system (FD)~\cite{vanhentenryck-ccfd-1998} where
variables are assumed to range over finite domains and, in addition to
equality, it is possible to have predicates (e.g., ``$\leq$'') that
restrict the values of a variable to some finite set.

An algebraic representation of $\CS$ is used in the present work.

\begin{definition}[Constraint Systems]
	\label{def:cs}
	A constraint system $(\CS)$ $\mathbf{C}$ is a complete algebraic
  lattice $(\Con, \sqsubseteq)$. The elements of $\Con$ are called
  \emph{constraints}. The symbols $\sqcup$, $\true$, and $\false$ are
  used to denote the least upper bound $($lub$)$ operation, the
  bottom, and the top element of $\mathbf{C}$, respectively.
\end{definition}

\noindent In Definition~\ref{def:cs}, a $\CS$ is characterized as a
\emph{complete algebraic lattice}. The elements of the lattice, the
\emph{constraints}, represent (partial) information. A constraint $c$
can be viewed as an \emph{assertion} (or a \emph{proposition}). The
lattice order $\sqsubseteq$ is meant to capture entailment of
information: $\cleq{d}{c}$, alternatively written $\cgeq{c}{d}$, means
that the assertion $c$ represents as much information as $d$.  Thus
$\cleq{d}{c}$ may be interpreted as saying that $c\entails d$ or that
$d$ can be \emph{derived} from $c$. The \emph{least upper bound
$($lub$)$} operator $\sqcup$ represents join of information and thus
$\join{c}{d}$ is the least element in the underlying lattice above $c$
and $d$, asserting that both $c$ and $d$ hold. The top element
represents the lub of all, possibly inconsistent, information, hence
it is referred to as $\false$. The bottom element $\true$ represents
the empty information.

\subsection{Order-sorted Rewriting Logic in a Nutshell}

This section presents an overview of rewriting logic and how it can be
used to obtain a rewriting logic semantics of a language. The reader
is referred to~\cite{meseguer-rltcs-1992,clavel-maudebook-2007}
and~\cite{meseguer-rlsproject-2013,rusu-rewt-2016}, respectively, for 
an in-depth treatment of these topics.

Rewriting logic~\cite{meseguer-rltcs-1992} is a general semantic framework
that unifies a wide range of models of concurrency. Language
specifications can be executed in Maude~\cite{clavel-maudebook-2007},
a high-performance rewriting logic implementation and benefit from a
wide set of formal analysis tools available to it, such as an LTL
model checker and an inductive theorem prover.

\subsubsection{Rewriting Logic.}
A {\em rewriting logic specification} or {\em rewrite theory} is a tuple
$\rcal = (\Sigma, E \uplus B, R)$ where:
\begin{itemize}
\item $(\Sigma, E\uplus B)$ is an order-sorted equational theory with
  $\Sigma = (S,\leq,F)$ a signature with finite poset of sorts
  $(S,\leq)$ and a set of function symbols $F$ typed with sorts in
  $S$; $E$ is a set of $\Sigma$-equations, which are universally
  quantified Horn clauses with atoms that are $\Sigma$-equations $t=u$
  with $t,u$ terms of the same sort; $B$ is a set of structural axioms
  (e.g., associativity, commutativity, identity) such that there
  exists a matching algorithm modulo $B$ producing a finite number of
  $B$-matching substitutions or failing otherwise; and
\item $R$ a set of universally quantified conditional rewrite rules
  of the form 
  \[\crl{t}{u}{\bigwedge_i \phi_i}\]
  where $t,u$ are $\Sigma$-terms of the same sort and each $\phi_i$ is
  a $\Sigma$-equality.
\end{itemize}
Given $X = \{X_s\}_{s\in S}$, an $S$-indexed family of disjoint
variable sets with each $X_s$ countably infinite, the {\em set of
  terms of sort $s$} and the {\em set of ground terms of sort $s$} are
denoted, respectively, by $T_\Sigma(X)_s$ and $T_{\Sigma,s}$;
similarly, $T_\Sigma(X)$ and $T_\Sigma$ denote, respectively, the set
of terms and the set of ground terms. The expressions
$\tcal_\Sigma(X)$ and $\tcal_\Sigma$ denote the corresponding
order-sorted $\Sigma$-term algebras.  All order-sorted signatures are
assumed {\em preregular}~\cite{goguen-ordsortalg1-1992}, i.e., each
$\Sigma$-term $t$ has a unique {\em least sort} $\ls{t} \in S$ s.t. $t
\in T_\Sigma(X)_{\ls{t}}$. It is also assumed that $\Sigma$ has {\em
  nonempty sorts}, i.e., $T_{\Sigma,s}\neq \emptyset$ for each $s\in
S$. Many-sorted equational logic is the special case of order-sorted
equational logic when the subsort relation $\leq$ is restricted to be
the identity relation over the sorts.

An equational theory $\ecal = (\Sigma,E)$ induces the congruence
relation $=_\ecal$ on $T_\Sigma(X)$ (or simply $=_E$) defined for $t,u
\in T_\Sigma(X)$ by $t =_\ecal u$ if and only if $\ecal \ded
\eq{t}{u}$, where $\ecal \ded \eq{t}{u}$ denotes $\ecal$-provability
by the deduction rules for order-sorted equational logic
in~\cite{meseguer-membership-1998}. For the purpose of this paper, such 
inference rules, which are analogous to those of many-sorted equational 
logic, are even simpler thanks to the assumption that $\Sigma$ has nonempty
sorts, which makes unnecessary the explicit treatment of universal
quantifiers. The expressions $\tcal_{\ecal}(X)$ and $\tcal_\ecal$
(also written $\tcal_{\Sigma/E}(X)$ and $\tcal_{\Sigma/E}$) denote the
quotient algebras induced by $=_\ecal$ on the term algebras
$\tcal_\Sigma(X)$ and $\tcal_\Sigma$, respectively;
$\tcal_{\Sigma/E}$ is called the {\em initial algebra} of
$(\Sigma,E)$.

A \textit{(topmost) rewrite theory} is a tuple $\rcal = (\Sigma,E,R)$
with top sort $\states$, i.e., no operator in $\Sigma$ has $\states$
as argument sort and each rule $\ccrl{l}{r}{\phi}\in R$ satisfies $l,r
\in T_\Sigma(X)_\states$ and $l \notin X$.  A rewrite theory $\rcal$
induces a rewrite relation $\rews_{\rcal}$ on $T_{\Sigma}(X)$ defined
for every $t,u \in T_\Sigma(X)$ by $t \rews_\rcal u$ if and only if
there is a rule $(\ccrl{l}{r}{\phi}) \in R$ and a substitution
$\func{\theta}{X}{T_\Sigma(X)}$ satisfying $t =_E l\theta$, $u =_E
r\theta$, and $E \vdash \phi\theta$. The tuple $\tcal_\rcal =
(\tcal_{\Sigma/E},\rews_\rcal^*)$ is called the {\em initial
  reachability model of $\rcal$}~\cite{bruni-semantics-2006}.

\subsubsection{Admissible Rewrite Theories.}
Appropriate requirements are needed to make an equational theory
$\ecal$ {\em admissible}, i.e., {\em executable} in rewriting
languages such as Maude~\cite{clavel-maudebook-2007}.  In this paper, it is
assumed that the equations $E$ can be oriented into a set of (possibly
conditional) {sort-decreasing}, {operationally terminating}, and
{confluent} rewrite rules $\overrightarrow{E}$ modulo $B$. The rewrite
system $\overrightarrow{E}$ is {\em sort decreasing} modulo $B$ if and
only if for each $(\crl{t}{u}{\cond}) \in \overrightarrow{E}$ and
substitution $\theta$, $\ls{t\theta} \geq \ls{u\theta}$ if
$(\Sigma,B,\overrightarrow{E}) \ded \cond\theta$.  The system
$\overrightarrow{E}$ is {\em operationally terminating} modulo
$B$~\cite{duran-operterm-2008} if and only if there is no infinite
well-formed proof tree in $(\Sigma,B,\overrightarrow{E})$
(see~\cite{lucas-ordersorted-2009} for terminology and details).  
Furthermore, $\overrightarrow{E}$ is {\em confluent} modulo $B$ if 
and only if for all $t,t_1,t_2 \in T_\Sigma(X)$, if $t \rews^*_{E/B} t_1$ 
and $t \rews^*_{E/B} t_2$, then there is $u \in T_\Sigma(X)$ such that $t_1
\rews^*_{E/B} u$ and $t_2 \rews^*_{E/B} u$.  The term $\can{t}{E/B}
\in T_\Sigma(X)$ denotes the {\em $E$-canonical form} of $t$ modulo
$B$ so that $t \rews_{E/B}^* \can{t}{E/B}$ and $\can{t}{E / B}$ cannot
be further reduced by $\rews_{E/B}$. Under sort-decreasingness,
operational termination, and confluence, the term $\can{t}{E/B}$ is
unique up to $B$-equality.

For a rewrite theory $\rcal$, the rewrite relation $\rews_\rcal$ is
undecidable in general, even if its underlying equational theory is
admissible, unless conditions such as
coherence~\cite{viry-coherence-2002} are given (i.e, whenever
rewriting with $\rews_{R/E \cup B}$ can be decomposed into rewriting
with $\rews_{E/B}$ and $\rews_{R/B}$). A key goal
of~\cite{rocha-rewsmtjlamp-2017} was to make such a relation both
decidable and symbolically executable when $\rcal$ is topmost and $E$
decomposes as $E_\lang \uplus B_\nlang$, representing a built-in
theory for which formula satisfiability is decidable and $B_\nlang$
has a matching algorithm.

\subsubsection{Rewriting Logic Semantics.}
The rewriting logic semantics of a language $\lcal$ is a rewrite
theory $\rcal_\lcal = (\Sigma_\lcal,E_\lcal \uplus B_\lcal, R_\lcal)$
where $\rews_{\rcal_\lcal}$ provides a step-by-step formal description
of $\lcal$'s {\em observable} run-to-completion mechanisms. The
conceptual distinction between equations and rules in $\rcal_\lcal$
has important consequences that are captured by rewriting logic's {\em
  abstraction dial}~\cite{meseguer-rlsproject-2013}. Setting the level of
abstraction in which all the interleaving behavior of evaluations in
$\lcal$ is observable, corresponds to the special case in which the
dial is turned down to its minimum position by having $E_\lcal \uplus
B_\lcal = \emptyset$. The abstraction dial can also be turned up to
its maximal position as the special case in which $R_\lcal =
\emptyset$, thus obtaining an equational semantics of $\lcal$ without
observable transitions. The rewriting logic semantics presented in
this paper is {\em faithful} in the sense that such an abstraction
dial is set at a position that exactly captures the interleaving
behavior of the concurrency model.

\subsubsection{Maude.}
Maude~\cite{clavel-maudebook-2007} is a language and system based on rewriting 
logic.
It supports order-sorted equational and rewrite theory specifications
in \emph{functional} and \emph{system} modules,
respectively. Admissibility of functional and system modules can be
checked with the help of the \emph{Maude Formal Environment}
(MFE)~\cite{duran-mfetalcott-2011,duran-mfecalco-2011}, an executable 
formal specification in Maude with tools to mechanically verify such
properties.  The MFE includes the Maude Termination Tool, the Maude
Sufficient Completeness Checker, the Church-Rosser Checker, the
Coherence Checker, and the Maude Inductive Theorem Prover.

\subsection{SMT-Solving}

Satisfiability Modulo Theories (SMT) studies methods for checking
satisfiability of first-order formulas in specific models. The SMT
problem is a decision problem for logical formulas with respect to
combinations of background theories expressed in classical first-order
logic with equality. An SMT instance is a formula $\phi$ (typically
quantifier free, but not necessarily) in first-order logic and a model
$\tcal$, with the goal of determining if $\phi$ is satisfiable in
$\tcal$.

In this work, the representation of the constraint system is based on
SMT solving technology. Given an many-sorted equational theory
$\ecal_\lang = (\Sigma_\lang,E_\lang)$ and a set of variables $X_\lang
\subseteq X$ over the sorts in $\Sigma_\lang$, the formulas under
consideration are in the set $\oqff{\Sigma_\lang}{X_\lang}$ of
quantifier-free $\Sigma_\lang$-formulas: each formula being a Boolean
combination of $\Sigma_\lang$-equation with variables in $X_\lang$
(i.e., atoms). The terms in $T_{\ecal_\lang}$ are called
\textit{built-ins} and represent the portion of the specification that
will be handled by the SMT solver (i.e., semantic data types). In this
setting, an SMT instance is a formula $\phi \in
\oqff{\Sigma_\lang}{X_\lang}$ and the initial algebra
$\tcal_{\ecal_\lang^{+}}$, where $\ecal_\lang^{+}$ is a
\textit{decidable extension} of $\ecal_\lang$ such that
\begin{align*}
  \phi
  \textnormal{ is satisfiable in $\tcal_{\ecal_\lang^+}$} \; \iff \; (\exists
  \func{\sigma}{X_\lang}{T_{\Sigma_\lang}})\; \tcal_{\ecal_\lang}\MODELS\phi\sigma.
\end{align*}
Many decidable theories $\ecal_{\lang}^+$ of interest are supported by
SMT solvers satisfying this requirement
(see~\cite{rocha-rewsmtjlamp-2017} for details).  In this work, the
latest alpha release of Maude that integrates
Yices2~\cite{dutertre-yices2-2014} and CVC4~\cite{barrett-cvc4-2011}
is used for reachability analysis.
\section{Spatial Concurrent Constraint Systems with Extrusion}
\label{sec.sccp}

This section presents the syntax and the structural operational
semantics of spatial concurrent constraints systems with extrusion,
which is based mainly on~\cite{knight-sccp-2012}. This section also
introduces an example to illustrate the main features of the language.

\subsection{Spatial Constraint Systems}

The authors of \cite{knight-sccp-2012} extended the notion of $\CS$ to
account for distributed and multi-agent scenarios where agents have
their own space for local information and computation.

\paragraph{Locality and Nested Spaces.}
Each agent $i$ has a \emph{space} function $\sfunc{\cdot}_i$ from
constraints to constraints (recall that constraints can be viewed as
assertions). The function
\begin{align*}
  \sfunc{c}_i
\end{align*}
can be interpreted as an assertion stating that $c$ is a piece of
information that resides \emph{within a space attributed to agent}
$i$. An alternative \emph{epistemic interpretation} of $\sfunc{c}_i$
is an assertion stating that agent $i$ \emph{believes} $c$ or that $c$
holds within the space of agent $i$ (but it may or may not hold
elsewhere). Both interpretations convey the idea that $c$ is local to
agent $i$. Following this intuition, the assertion
\begin{align*}
  \sfunc{\sfunc{c}_j}_i
\end{align*}
is a hierarchical spatial specification stating that $c$ holds within
the local space the agent $i$ attributes to agent $j$. Nesting of
spaces such as in $\sfunc{\sfunc{\ldots\sfunc{c}_{i_m}
    \ldots}_{i_2}}_{i_1}$ can be of any depth.

\paragraph{Parallel Spaces.}
A constraint of the form
\begin{align*}
  \join{\sfunc{c}_i}{\sfunc{d}_j}
\end{align*}
can be seen as an assertion specifying that $c$ and $d$ hold within
two \emph{parallel/neighboring} spaces that belong to agents $i$ and
$j$. From a computational/concurrency point of view, it is possible to
think of $\sqcup$ as parallel composition; from a logic point of view,
$\sqcup$ corresponds to conjunction.

The notion of an $n$-agent spatial constraint system is formalized in
Definition~\ref{def.sccp.scs}.

\index{$({\Con},\sqsubseteq,\sfuncs)$, spatial constraint system, 
	$\SCS$}
\index{$\sfunc{\cdot}$, space function}
\begin{definition}[Spatial Constraint System \cite{knight-sccp-2012}]
  \label{def.sccp.scs}
  An $n$-agent \emph{spatial constraint system $($$n$-$\SCS$$)$} ${\bf
    C}$ is a $\CS$ $(\Con, \sqsubseteq)$ equipped with $n$ self-maps
  $\sfuncs$ over its set of constraints ${\Con}$ satisfying for each
  function $\sfunc{\cdot}_i:\Con \rightarrow \Con$:
  \begin{enumerate} 
    \item  [S.1]
    $\sfunc{\true}_i = \true,$ \text{and}
    \item [S.2] $\sfunc{\join{c}{d}}_i = \join{\sfunc{c}_i}
    {\sfunc{d}_i} \ \ \text{ for each } c,d \in \Con.$
  \end{enumerate}
\end{definition}
Property S.1 in Definition \ref{def.sccp.scs} requires space functions
to be strict maps (i.e., bottom preserving) where an empty local space
amounts to having no knowledge. Property S.2 states that space
functions preserve (finite) lubs, and also allows to join and
distribute the local information of any agent $i.$ Henceforth, given
an $n$-$\SCS$ ${\bf C}$, each $\sfunc{\cdot}_i$ is thought as the
\emph{space} (or space function) of the agent $i$ in ${\bf C}$. The
tuple $({\Con},\sqsubseteq,\sfuncs)$ denotes the corresponding
$n$-$\SCS$ with space functions $\sfuncs.$ Components of an $n$-$\SCS$
tuple shall be omitted when they are unnecessary or clear from the
context. When $n$ is unimportant, $n$-$\SCS$ is simply written as
$\SCS$.

\paragraph{Extrusion.}
Extrusion (i.e., mobility) plays a key role in distributed systems.
Following the algebraic approach, it is possible to provide each agent
$i$ with an \emph{extrusion} function $\uparrow_i : Con \rightarrow
Con$ ~\cite{guzman-sccpe-2017,guzman-sccp-2016}. The process
$\extr{i}{c}$ within a space context $\sfunc{\cdot}_i$ means that the
process $c$ must be executed outside of agent's $i$ space.

%
Definition~\ref{def.sccp.scse} presents the extension of spatial
constraint systems with extrusion.

\begin{definition}[Spatial Constraint System with 
  Extrusion~\cite{guzman-sccpe-2017}]
  \label{def.sccp.scse}
  An \emph{$n$-agent spatial constraint system with extrusion
    ($n$-$\SCSE$) is an $n$-$\SCS$ ${\bf C}$ equipped with $n$
    self-maps $\uparrow_1, \dots , \uparrow_n$ over $\Con$, written
    $({\bf C}, \uparrow_1, \dots , \uparrow_n)$}, such that each
  $\uparrow_i$ is the right inverse of $\sfunc{\cdot}_i$.
\end{definition}


\subsection{Spatial $\CCP$ with Extrusion}

The \emph{spatial concurrent constraint programming with extrusion}
($\SCCPwE$) calculus presented in this section follows the
developments of~\cite{knight-sccp-2012} and~\cite{guzman-sccpe-2017}.
The syntax of $\SCCPwE$ is parametric on an $\SCSE$ and it is
presented in Definition~\ref{def.sccp.sccpe.syntax}.


\begin{definition}[$\SCCPwE$\ Processes]\label{def.sccp.sccpe.syntax}
  Let ${\bf C}=({\Con},\sqsubseteq)$ be a constraint system, $A$ a
  set of $n$ agents, and $V$ an infinite countable set of
  variables. Let $(\textbf{C}, \sfuncs, \efuncs)$ be an $n$-$\SCSE$
  and consider the following EBNF-like syntax:
\vspace{-0.15cm}
\begin{align*} P & \; ::= \;
	\Stop\;\mid\; 
	\tell(c) \;\mid\; 
	\ask(c)\rightarrow P \;\mid\; 
	P \parallel P \;\mid\; 
	\K i P \;\mid\; 
	\extr{i}{P} \;\mid\;
	x \;\mid\; 
	\mu x.P 
\end{align*}
  \noindent where $c \in \Con$, $i \in A$, and $x \in V$. An expression
  $P$ in the above syntax is a \emph{process} if and only if every
  variable $x$ in $P$ occurs in the scope of an expression of the form
  $\mu x.P$. The set of processes of $\SCCPwE$ is denoted by
  $\textit{Proc}$.
\end{definition}

\noindent The $\SCCPwE$ calculus can be thought of as a
\emph{shared-spaces} model of computation. Each agent $i \in A$ has a
computational space of the form ${\K i \cdot}$ possibly containing
processes and other agents' spaces.  The basic constructs of $\SCCPwE$
are tell, ask, and parallel composition, and they are defined as in
standard $\CCP$~\cite{saraswat-ccpsem-1991}. A process $\tell(c)$
running in an agent $i \in A$ adds $c$ to its local store $s_i$,
making it available to other processes in the same space. This
addition, represented as $s_i \sqcup c$, is performed even if the
resulting constraint is inconsistent.  The process $\ask(c)\rightarrow
P$ running in space $i$ may execute $P$ if $c$ is entailed by $s_i$,
i.e., $c \sqsubseteq s_i$. The process $P \parallel Q$ specifies the
{\it parallel execution} of processes $P$ and $Q$; given $I=\{
i_1,\ldots,i_m \} \subseteq A$, the expression $\prod_{i \in I} P_i$
is used as a shorthand for $P_{i_1} \parallel \ldots \parallel
P_{i_m}$.  A construction of the form ${\K iP}$ denotes a process $P$
running within the agent $i$'s space. Any information that $P$
produces is available to processes that lie within the same space. The
process $\extr{i}{P}$ denotes that process $P$ runs outside the space
of agent $i$ and the information posted by $P$ resides in the store of
the parent of agent $i$. The behavior of recursive definitions of the
form $\mu x.P$ is represented by $P[\mu x.P / x]$, i.e., every free
occurrence of $x$ in $P$ is replaced with $\mu x.P$. In order to make
recursive definitions finite, it is necessary for recursion to be
guarded by an \emph{ask}, i.e., every occurrence of $x$ in $P$ is
within the scope of an ask process.  Note that the process $P$ in a
construction of the form $\ask(true) \rightarrow P$ is unguarded.

\begin{example}\label{exa.sccp.spatial}
  Consider the processes $P = \tell(c)$ and $Q=\ask(c)\rightarrow
  \tell(d)$. Consider ${\K i P } \parallel {\K i {Q} }$: by the above
  intuitions the constraints $c$ and $d$ are added to store of agent
  $i$. A similar behavior is achieved by the process ${\K i {P
      \parallel Q}}$, which also produces $c \sqcup d$ in the store of
  agent $i$ (note that $\sfunc{c \sqcup d}_i$ is equivalent to
  $\sfunc{c}_i\sqcup \sfunc{d}_i$ by Property S.2 in Definition
  \ref{def.sccp.scs}). In contrast, the process ${\K j {P} } \parallel
      {\K i {Q}}$ with $i\neq j$, does not necessarily add $d$ to the
      space of agent $i$ because $c$ is not made available for agent
      $i$; likewise in $P \parallel {\K i {Q}}$, $d$ is not added to
      the space of agent $i$.  Finally, consider ${\K i {P \parallel
          {\K j {\extr{j}{Q}}}}}$. In this case, because of extrusion,
      both $c$ and $d$ will be added to store of agent $i$. However,
      ${\K j {\K i {P \parallel {\extr{i}{Q}}}}}$ with $i\neq j$,
      results in the space of agent $i$ having $c$ and $d$, but
      neither are available for agent $j$.  Note that in ${\K i P}
      \parallel {\K j {\extr{i}{Q}}}$, the constraint $c$ is added to
      the space of agent $i$, but since $Q$ cannot be extruded in ${\K
        j {\extr{i}{Q}}}$, $d$ is not added to the spaces of $i$ or
      $j$.
\end{example}


Next, the notion of the projection of a spatial constraint $c$ for an
agent $i$ is introduced.

\begin{definition}[Views]\label{def.sccp.view} 
  The agent $i$'s \emph{view} of $c$, denoted as $\prj ci$, is given 
  by $\prj ci=\bigsqcup\{d\;|\;\cleq{\sfunc{d}_i}{c}\}$.
\end{definition}
The intuition is that $\prj ci$ represents all the information the
agent $i$ may see or have in $c$. For example if $c=\sfunc{d}_i \sqcup
\sfunc{e}_ j$ then agent $i$ sees $d$, so $d \sqsubseteq \prj ci$.

\subsection{Structural Operational Semantics of $\SCCPwE$}

The operational semantics of $\SCCPwE$ is defined over
configurations. A \textit{configuration} is a pair of the form
$\pairccp{P}{c} \in {\it Proc} \times \Con$, where $P$ is a process
and $c$ is the spatial distribution of information available to it;
the set of configurations is denoted by $\textit{Conf}$. The
structural operational semantics of $\SCCPwE$ is captured by the
binary transition relation $\rrarrow\;\subseteq{\it Conf}\times{\it
  Conf}$, defined by the rules in Figure \ref{fig.sccp.opsem}.

\begin{figure}
  \centering
  $
  \begin{array}{c}
    \infer[\rTell]{\conf{\tellp{c}}{d}  \redi 
	\conf{\Stop}{d \sqcup c}}{
    }
    \qquad
    \infer[\rAsk]{\conf{\askp{c}{P}}{d} \redi
    \conf{P}{d}} {\cleq{c}{d}}
    \\\\
    
    \infer[\rPar]{\conf{P\parallel Q}{d} \redi
    \conf{P'\parallel Q}{d'}} {\conf{P}{d} \redi \conf{P'}{d'}}
    \qquad
    \infer[\rSp]{\conf{\K i P}{c} \redi
    \conf{\K i {P'}} {c\sqcup \sfunc{c'}_i}} {\conf{P}{c^i} 
    \redi \conf{P'}{c'}}
    \\\\
    
    \infer[\rRec]{\conf{\mu x.P}{d} \redi \gamma} 
	{\conf{P \sfunc{\mu x.P / x }}{d} \redi \gamma}
	\qquad
	\infer[\rExt]{\conf{\extr{i}{\K i {P}}}{d} \redi 
	\conf{\extr{i}{\K i {P'}}}{d'}} 
	{\conf{P}{d} \redi \conf{P'}{d'}}
	{}
  \end{array}
  $
  \caption{Structural operational semantics of $\SCCPwE$.}
  \label{fig.sccp.opsem}
\end{figure}

The rules $\rTell, \rAsk$, $\rPar$, and $\rRec$ for the
basic processes and recursion are the standard ones in $\CCP$. In
order to avoid another version of the $\rPar$ rule for process $Q$,
parallel composition is assumed to be (associative and) commutative.
The rule $\rExt$ is a context-dependent definition for extrusion,
i.e., it requires a space process (i.e., $\sfunc{\cdot}_j$ with $j=i$)
and specifies extrusion in the sense explained before. In the rule
$\rSp$, $\prj ci$ represents all the information the agent $i$ may see
or have in $c$: namely, $P$ runs with store ${\prj ci}$, i.e., with
the agent's $i$ \emph{view} of $c$. Also note the information $c'$ added
to $\prj ci$ by the computation of $P$ corresponds to the information
added by $\K iP$ to the space of agent $i$.

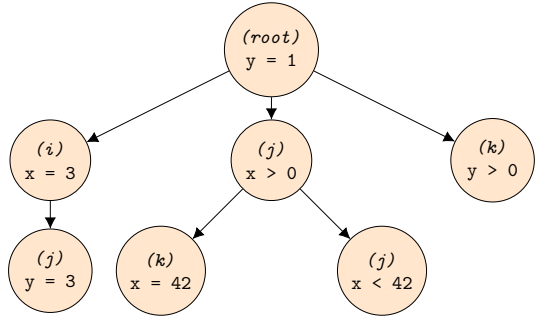
\begin{figure}[thbp] 
  \centering
	\resizebox{7cm}{!}{%
    \begin{tikzpicture}[node distance=2cm,
        every node/.style={fill=white, font=\sffamily}, align=center]
      \node (root)  [store]                
				    {\textit{(root)}\\y = 1};
            \node (i)		[store, below of=root, left of=root, xshift=-2cm]
			            {\textit{(i)}\\x = 3};
                  \node (j)  	[store, below of=root]  
				                {\textit{(j)}\\x > 0};
	                      \node (k)  	[store, below of=root, right of=root, xshift=2cm]  
					                    {\textit{(k)}\\y > 0};				   
                              \node (ji)   	[store, below of=i] 
			                              {\textit{(j)}\\y = 3};
                                    \node (kj)   	[store, below of=j, left of=j] 
				                                  {\textit{(k)}\\x = 42};
                                          \node (jj)   	[store, below of=j, right of=j] 
			                                          {\textit{(j)}\\x < 42};				  
                                                \draw[->]   (root) -- (i);
                                                \draw[->]   (root) -- (j);
                                                \draw[->]   (root) -- (k);     
                                                \draw[->]      (i) -- (ji);
                                                \draw[->]      (j) -- (kj);
                                                \draw[->]      (j) -- (jj);      
  \end{tikzpicture}}
  \caption{A spatial hierarchy of processes.}
  \label{fig.sccp.tree}
\end{figure}

\begin{example}
  One way of understanding $\rSp$ is by considering a constraint $c$
  as a tree-like structure. Each node in such a tree corresponds to
  the information (constraint) contained in an agent's space. Edges in
  the tree-like structure define the spatial hierarchy of agents. For
  example, given an appropriate $\SCSE$, Figure~\ref{fig.sccp.tree}
  corresponds to the following store $d$: \[ d \defsymbol (y=1) \sqcup
  \sfunc{x=3 \sqcup\sfunc{y=3}_j}_i \sqcup \sfunc{x>0
    \sqcup\sfunc{x=42}_k \sqcup\sfunc{x<42}_j}_j
  \sqcup\sfunc{y>0}_k. \] The configuration
  \[\conf{\sfunc{\askp{x=42}{P}}_j}{d}\] is a deadlock, while
  \[\conf{\sfunc{\sfunc{\askp{x=42}{P}}_k}_j}{d}
  \redi{\sfunc{\sfunc{P}_k}_j}{d}\] by the rules $\rSp$ and $\rAsk$,
  using the fact that $\cleq{x=42}{\prj{{\prj dj}}k}$.
\end{example}

\subsection{Example}
\label{sec.sccp.example}

This section includes an example that will be used throughout the
paper. The example models a system where an agent sends messages to
other agents through the spatial hierarchy.

As mentioned earlier, a distinctive property about $\SCCPwE$ is that
it allows to have inconsistent information within spaces; i.e., one
agent may have local information $c$ and the other some local
information $d$ such that $\join {c}{d} = \false$ (e.g., $c$ and $d$
could be $x=25$ and $x<20$, respectively). This also means that an
agent can send inconsistent information to different agents. 
\begin{figure}[htbp] 
	\centering
	\begin{subfigure}[b]{.32\textwidth}
		\resizebox{3cm}{!}{%
			\begin{tikzpicture}[node distance=2cm, 
			every node/.style={fill=white, font=\sffamily}, align=center]
			\node (root)  [curr]                
			{\textit{(root)}\\true\\\rule{1.5cm}{1.3pt}\\$P$};
			\node (0)	  [store, below of=root, left of=root]        
			{\textit{(0)}\\X = 25};
			\node (1)  	  [store, below of=root, right of=root]  
			{\textit{(1)}\\true};
			\node (01)    [store, below of=1] 
			{\textit{(0)}\\Y < 5};
			\draw[->]   (root) -- (0);
			\draw[->]   (root) -- (1);
			\draw[->]      (1) -- (01);
			\end{tikzpicture}}
		\caption{Initial state.}
		\label{fig.sccp.sccp-p}
	\end{subfigure}
	\begin{subfigure}[b]{.32\textwidth}
		\centering
		\resizebox{3cm}{!}{%
			\begin{tikzpicture}[node distance=2cm, 
			every node/.style={fill=white, font=\sffamily}, align=center]
			\node (root)  [store]                
			{\textit{(root)}\\true};
			\node (0)	  [curr, below of=root, left of=root]        
			{\textit{(0)}\\X = 25\\\rule{1.5cm}{1.3pt}\\$\extr{0}{\sfunc{{P_1}\parallel {\sfunc{P_2 }_0}}_1}$};
			\node (1)  	  [store, below of=root, right of=root]  
			{\textit{(1)}\\true};
			\node (01)    [store, below of=1] 
			{\textit{(0)}\\Y < 5};
			\draw[->]   (root) -- (0);
			\draw[->]   (root) -- (1);
			\draw[->]      (1) -- (01);
			\end{tikzpicture}}
		\caption{Execution in space 0.}
		\label{fig.sccp.sccp-p0}
	\end{subfigure}
	\begin{subfigure}[b]{.32\textwidth}
		\centering
		\resizebox{3cm}{!}{%
			\begin{tikzpicture}[node distance=2cm,
			every node/.style={fill=white, font=\sffamily}, align=center]
			\node (root)  [curr]                
			{\textit{(root)}\\true\\\rule{1.5cm}{1.3pt}\\$\sfunc{{P_1}\parallel {\sfunc{P_2 }_0}}_1$};
			\node (0)	  [store, below of=root, left of=root]        
			{\textit{(0)}\\X = 25};
			\node (1)  	  [store, below of=root, right of=root]  
			{\textit{(1)}\\true};
			\node (01)    [store, below of=1] 
			{\textit{(0)}\\Y < 5};
			\draw[->]   (root) -- (0);
			\draw[->]   (root) -- (1);
			\draw[->]      (1) -- (01);
			\end{tikzpicture}}
		\caption{Extrusion from space 0.}
		\label{fig.sccp.sccp-p1}
	\end{subfigure}	
	
	\begin{subfigure}[b]{.32\textwidth}
		\centering
		\resizebox{3cm}{!}{%
			\begin{tikzpicture}[node distance=2cm, 
			every node/.style={fill=white, font=\sffamily}, align=center]
			\node (root)  [store]                
			{\textit{(root)}\\true};
			\node (0)	  [store, below of=root, left of=root]        
			{\textit{(0)}\\X = 25};
			\node (1)  	  [curr, below of=root, right of=root]  
			{\textit{(1)}\\true\\\rule{1.5cm}{1.3pt}\\${P_1}\parallel {\sfunc{P_2 }_0}$};
			\node (01)    [store, below of=1, yshift=-0.5cm] 
			{\textit{(0)}\\Y < 5};
			\draw[->]   (root) -- (0);
			\draw[->]   (root) -- (1);
			\draw[->]      (1) -- (01);
			\end{tikzpicture}}
		\caption{Execution in space 1.}
		\label{fig.sccp.sccp-p2}
	\end{subfigure}
	\begin{subfigure}[b]{.32\textwidth}
		\centering
		\resizebox{3cm}{!}{%
			\begin{tikzpicture}[node distance=2cm, 
			every node/.style={fill=white, font=\sffamily}, align=center]
			\node (root)  [store]                
			{\textit{(root)}\\true};
			\node (0)	  [store, below of=root, left of=root]        
			{\textit{(0)}\\X = 25};
			\node (1)  	  [curr, below of=root, right of=root]  
			{\textit{(1)}\\true\\\rule{1.5cm}{1.3pt}\\$\tell(Z \geq 10)$};
			\node (01)    [curr, below of=1, yshift=-0.85cm] 
			{\textit{(0)}\\Y<5\\\rule{1.5cm}{1.3pt}\\$P_2$};
			\draw[->]   (root) -- (0);
			\draw[->]   (root) -- (1);
			\draw[->]      (1) -- (01);
			\end{tikzpicture}}
		\caption{Parallel execution.}
		\label{fig.sccp.sccp-p3}
	\end{subfigure}
	\begin{subfigure}[b]{.32\textwidth}
		\centering
		\resizebox{3cm}{!}{%
			\begin{tikzpicture}[node distance=2cm,
			every node/.style={fill=white, font=\sffamily}, align=center]
			\node (root)  [store]                
			{\textit{(root)}\\true};
			\node (0)	  [store, below of=root, left of=root]        
			{\textit{(0)}\\X = 25};
			\node (1)  	  [store, below of=root, right of=root]  
			{\textit{(1)}\\Z >= 10};
			\node (01)    [curr, below of=1, yshift=-0.5cm] 
			{\textit{(0)}\\Y<5\\\rule{1.5cm}{1.3pt}\\$P_2$};
			\draw[->]   (root) -- (0);
			\draw[->]   (root) -- (1);
			\draw[->]      (1) -- (01);
			\end{tikzpicture}}
		\caption{Tell execution.}
		\label{fig.sccp.sccp-p4}
	\end{subfigure}
	
	\begin{subfigure}[b]{.32\textwidth}
		\centering
		\resizebox{3cm}{!}{%
			\begin{tikzpicture}[node distance=2cm, 
			every node/.style={fill=white, font=\sffamily}, align=center]
			\node (root)  [store]                
			{\textit{(root)}\\true};
			\node (0)	  [store, below of=root, left of=root]        
			{\textit{(0)}\\X = 25};
			\node (1)  	  [store, below of=root, right of=root]  
			{\textit{(1)}\\Z >= 10};
			\node (01)    [curr, below of=1, yshift=-0.75cm] 
			{\textit{(0)}\\Y < 5\\\rule{1.5cm}{1.3pt}\\${\extr{0}{\extr{1}
						{\sfunc{\sfunc{P_3}_2}_0}}}$};
			\draw[->]   (root) -- (0);
			\draw[->]   (root) -- (1);
			\draw[->]      (1) -- (01);
			\end{tikzpicture}}
		\caption{Ask execution.}
		\label{fig.sccp.sccp-p5}
	\end{subfigure}
	\begin{subfigure}[b]{.32\textwidth}
		\centering
		\resizebox{3cm}{!}{%
			\begin{tikzpicture}[node distance=2cm,
			every node/.style={fill=white, font=\sffamily}, align=center]
			\node (root)  [store]                
			{\textit{(root)}\\true};
			\node (0)	  [store, below of=root, left of=root]        
			{\textit{(0)}\\X = 25};
			\node (1)  	  [curr, below of=root, right of=root]  
			{\textit{(1)}\\Z >= 10\\\rule{1.5cm}{1.3pt}\\$\extr{1}
				{\sfunc{\sfunc{P_3}_2}_0}$};
			\node (01)    [store, below of=1, yshift=-0.5cm] 
			{\textit{(0)}\\Y < 5};		
			\draw[->]   (root) -- (0);
			\draw[->]   (root) -- (1);
			\draw[->]      (1) -- (01);
			\end{tikzpicture}}
		\caption{Extrusion from space 0.}
		\label{fig.sccp.sccp-p6}
	\end{subfigure}
	\begin{subfigure}[b]{.32\textwidth}
		\centering
		\resizebox{3cm}{!}{%
			\begin{tikzpicture}[node distance=2cm, 
			every node/.style={fill=white, font=\sffamily}, align=center]
			\node (root)  [curr] 
			{\textit{(root)}\\true\\\rule{1.5cm}{1.3pt}\\$\sfunc{\sfunc{P_3}_2}_0$};
			\node (0)	  [store, below of=root, left of=root]        
			{\textit{(0)}\\X = 25};
			\node (1)  	  [store, below of=root, right of=root]  
			{\textit{(1)}\\Z >= 10};
			\node (01)    [store, below of=1] 
			{\textit{(0)}\\Y < 5};
			\draw[->]   (root) -- (0);
			\draw[->]   (root) -- (1);
			\draw[->]      (1) -- (01);
			\end{tikzpicture}}
		\caption{Extrusion from space 1.}
		\label{fig.sccp.sccp-p7}
	\end{subfigure}
	
	\begin{subfigure}[b]{.32\textwidth}
		\centering
		\resizebox{3.5cm}{!}{%
			\begin{tikzpicture}[node distance=2cm,
			every node/.style={fill=white, font=\sffamily}, align=center]
			\node (root)  [store]                
			{\textit{(root)}\\true};
			\node (0)	  [curr, below of=root, left of=root]        
			{\textit{(0)}\\X = 25\\\rule{1.5cm}{1.3pt}\\$\sfunc{P_3}_2$};
			\node (1)  	  [store, below of=root, right of=root]  
			{\textit{(1)}\\Z >= 10};
			\node (01)    [store, below of=1] 
			{\textit{(0)}\\Y < 5};
			\draw[->]   (root) -- (0);
			\draw[->]   (root) -- (1);
			\draw[->]      (1) -- (01);
			\end{tikzpicture}}
		\caption{Execution in space 0.}
		\label{fig.sccp.sccp-p8}
	\end{subfigure}
	\begin{subfigure}[b]{.32\textwidth}
		\centering
		\resizebox{3.5cm}{!}{%
			\begin{tikzpicture}[node distance=2cm,
			every node/.style={fill=white, font=\sffamily}, align=center]
			\node (root)  [store]                
			{\textit{(root)}\\true};
			\node (0)	  [store, below of=root, left of=root]        
			{\textit{(0)}\\X = 25};
			\node (1)  	  [store, below of=root, right of=root]  
			{\textit{(1)}\\Z >= 10};
			\node (01)    [store, below of=1] 
			{\textit{(0)}\\Y < 5};
			\node (20)    [curr, below of=0, yshift=-0.5cm] 
			{\textit{(2)}\\true\\\rule{1.5cm}{1.3pt}\\$\tell(W<Y)$};			
			\draw[->]   (root) -- (0);
			\draw[->]   (root) -- (1);
			\draw[->]      (1) -- (01);
			\draw[->]      (0) -- (20);
			\end{tikzpicture}}
		\caption{Execution in space 2.}
		\label{fig.sccp.sccp-p9}
	\end{subfigure}	
	\begin{subfigure}[b]{.32\textwidth}
		\centering
		\resizebox{3.5cm}{!}{%
			\begin{tikzpicture}[node distance=2cm,
			every node/.style={fill=white, font=\sffamily}, align=center]
			\node (root)  [store]                
			{\textit{(root)}\\true};
			\node (0)	  [store, below of=root, left of=root]        
			{\textit{(0)}\\X = 25};
			\node (1)  	  [store, below of=root, right of=root]  
			{\textit{(1)}\\Z >= 10};
			\node (01)    [store, below of=1] 
			{\textit{(0)}\\Y < 5};
			\node (20)    [store, below of=0] 
			{\textit{(2)}\\W < Y};			
			\draw[->]   (root) -- (0);
			\draw[->]   (root) -- (1);
			\draw[->]      (1) -- (01);
			\draw[->]      (0) -- (20);
			\end{tikzpicture}}
		\caption{Tell execution (final state).}
		\label{fig.sccp.sccp-final}
	\end{subfigure}	
	\caption{Execution of process $P$ and evolution of the $\SCCPwE$ system.}
	\label{fig.sccp.sccp}
\end{figure}
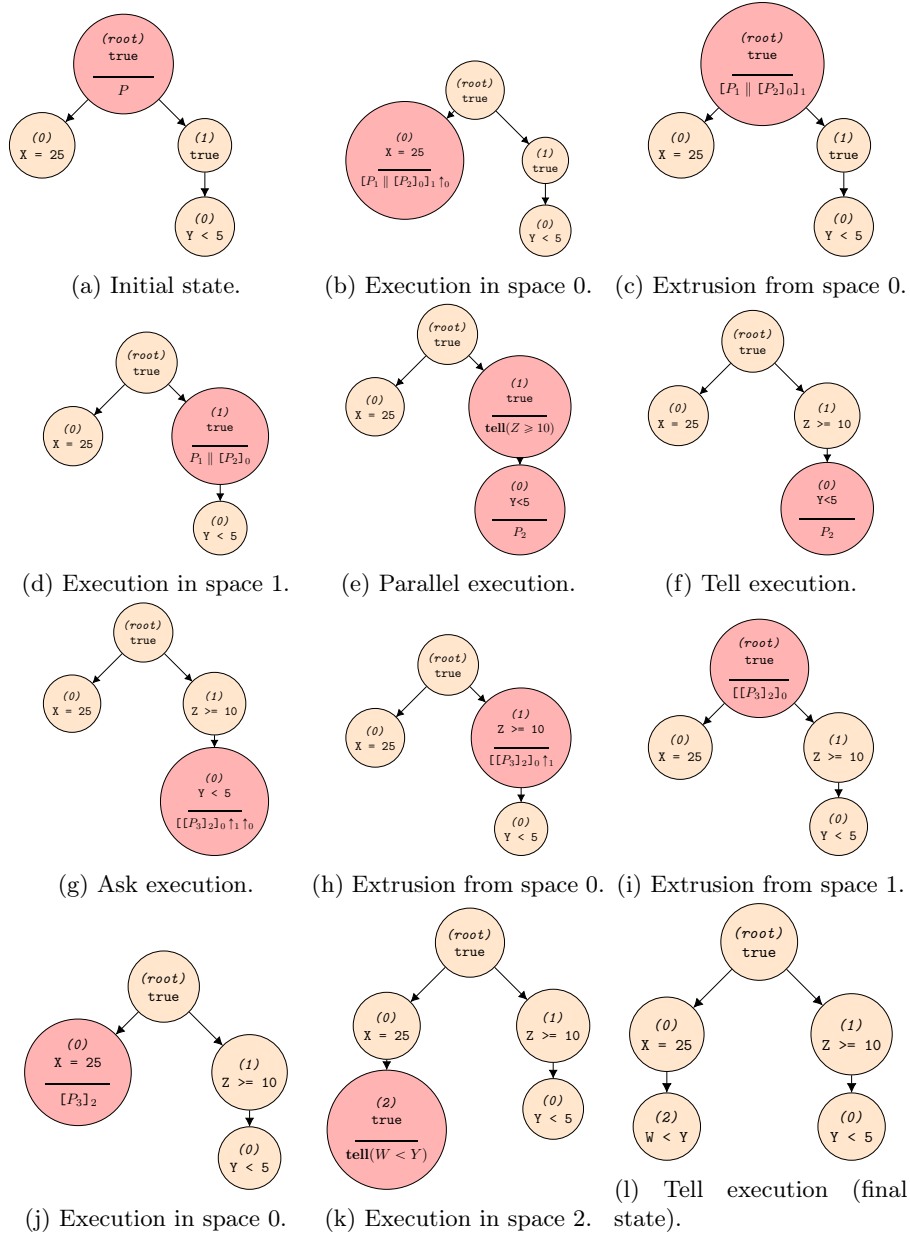

As an example of how hierarchical distributed processes evolve with
respect to the $\SCCPwE$'s structural operational semantics, consider
the sequence of system states depicted in Figure~\ref{fig.sccp.sccp}.
These states (figures~\ref{fig.sccp.sccp-p}-\ref{fig.sccp.sccp-final})
correspond to a step-by-step execution from the initial configuration
(Figure~\ref{fig.sccp.sccp-p}), with the process $P$ defined as
follows:
\begin{align*}
	P & \defsymbol
	\sfunc{\extr{0}{\sfunc{{P_1}\parallel 
				{\sfunc{P_2 }_0}}_1}}_0
\end{align*}
and where:
\begin{align*}
	P_1 & \defsymbol \tell(Z \geq 10) \\
	P_2 & \defsymbol \ask(Y<20)\rightarrow {\extr{0}{\extr{1}
			{\sfunc{\sfunc{P_3}_2}_0}}} \\
	P_3 & \defsymbol \tell(W<Y).
\end{align*}
\section{Rewriting Logic Semantics}
\label{sec.rew}

The rewriting logic semantics of a language $\lcal$ is a rewrite
theory $\rcal_\lcal = (\Sigma_\lcal,E_\lcal \uplus B_\lcal, R_\lcal)$
where $\rews_{\rcal_\lcal}$ provides a step-by-step formal description
of $\lcal$'s {\em observable} run-to-completion mechanisms. The
conceptual distinction between equations and rules in $\rcal_\lcal$
has important consequences that are captured by rewriting logic's {\em
	abstraction dial}~\cite{meseguer-rlsproject-2013}. Setting the level
of abstraction in which all the interleaving behavior of evaluations
in $\lcal$ is observable, corresponds to the special case in which the
dial is turned down to its minimum position by having $E_\lcal \uplus
B_\lcal = \emptyset$. The abstraction dial can also be turned up to
its maximal position as the special case in which $R_\lcal =
\emptyset$, thus obtaining an equational semantics of $\lcal$ without
observable transitions. The rewriting logic semantics $\rcal$ sets
such an abstraction dial at a position that exactly captures the
interleaving behavior of $\SCCPwE$.

This section presents a rewriting logic semantics for $\SCCPwE$ in
the form of a rewrite theory $\rcal = (\Sigma, E, R)$ with
\textit{topsort} \cde{Sys}. The data types supporting the state
structure are defined by the equational theory $(\Sigma, E)$ and the
state transitions are axiomatized by the rewrite rules $R$. The
constraint system, as detailed later in this section, is materialized
by an equational theory $(\Sigma_\lang,E_\lang) \subseteq (\Sigma,E)$
of built-ins whose quantifier-free formulas are handled by SMT
decision procedures. The complete specification of $\rcal$ can be
found in Appendix~\ref{sec.sccp.maude}.

Figure~\ref{fig.rew.modstr} depicts the module structure of $\rcal$.
Maude offers three modes for importing a module, namely,
$\cde{protecting}$, $\cde{extending}$, and $\cde{including}$. When a
functional module $M'$ protects a functional module $M$, it means that
the data types from $M$ are kept the same in $M'$ (i.e., no junk and
no confusion are added to the sorts of $M$). When a functional module
$M'$ extends a functional module $M$, it means that the data types
from $M$ can be extended with new terms, but existing terms are not
identified (i.e., junk is allowed but confusion is not).  When a
functional module $M'$ imports a functional module $M$, it means that
the data types from $M$ can be extended with new terms and existing
terms can be identified (i.e., junk and confusion are allowed). In
$\rcal$, modules are imported by protecting (denoted by a triple arrow
$\Rrightarrow$) or including (denoted by a single arrow $\rightarrow$)
submodules. See~\cite{clavel-maudebook-2007} for details about the
three different modes of module importation in Maude.


\begin{figure}[htb] 
  \centering
  \begin{tikzpicture}[node distance=1.35cm,
    every node/.style={fill=white, font=\sffamily}, align=center]
    \node (SCCP)    [module]                
                    {SCCP};
    \node (STATE)   [module, below of=SCCP]        
    				{SCCP-STATE};
    \node (SYNTAX)  [module, below of=STATE]  
    				{SCCP-SYNTAX};
    \node (AGENT)   [module, below of=SYNTAX, xshift=-2cm] 
    				{AGENT-ID};
    \node (INAT)    [module, below of=AGENT, xshift=-2cm]    
    				{INAT};
    \node (EXT)     [module, below of=AGENT, xshift=2cm]
    				{EXT-BOOL};
    \node (BOOL)    [module, below of=EXT]
    				{BOOL};
    \node (OPS)     [module, below of=BOOL] 
    				{BOOL-OPS};
    \node (TRUTH) 	[module, below of=OPS]
    				{TRUTH-VALUE};
    \node (SMT)     [module, below of=SCCP, right of=EXT, xshift=2cm]
    				{SMT-UTIL};
    \node (INTEGER) [module, below of=SMT, right of=BOOL, xshift=2cm]
    				{INTEGER};
    \node (BOOLEAN) [module, below of=INTEGER, right of=OPS, xshift=2cm]
    				{BOOLEAN};
    \node (META)    [module, below of=SMT, right of=SMT, xshift=2cm]
    				{META-LEVEL};
    \draw[triplearrow]   (SYNTAX) -- (STATE);
    \draw[triplearrow]    (AGENT) |- (SYNTAX);
    \draw[triplearrow]     (INAT) |- (AGENT);
    \draw[triplearrow]      (EXT) |- (AGENT);
    \draw[triplearrow]     (BOOL) -- (EXT);
    \draw[triplearrow]      (OPS) -- (BOOL);
    \draw[triplearrow]    (TRUTH) -- (OPS);
    \draw[triplearrow]    (TRUTH) -| (INAT);
    \draw[triplearrow]  (BOOLEAN) -- (INTEGER);
    \draw[triplearrow]  (INTEGER) -- +(-1.7,0) |- (SYNTAX);
    \draw[triplearrow]      (SMT) |- (SCCP);
    \draw[triplearrow]     (META) |- (SMT);
    \draw[->]     (STATE) -- (SCCP);
    \draw[->]   (INTEGER) -- (SMT);    
    \draw[dashed]          (META) -- +(0,-2);
  \end{tikzpicture}
  \caption{Module structure of the rewriting logic semantics of $\SCCPwE$.}
  \label{fig.rew.modstr}
\end{figure}

\subsection{The Constraint System}

The materialization of the constraint system in $\rcal$ uses SMT
solving technology. Given a \textit{many-sorted} (i.e., order-sorted
without sort structure) equational theory $\ecal_\lang =
(\Sigma_\lang,E_\lang)$ and a set of variables $X_\lang \subseteq X$
over the sorts in $\Sigma_\lang$, the formulas under consideration are
in the set $\oqff{\Sigma_\lang}{X_\lang}$ of quantifier-free
$\Sigma_\lang$-formulas: each formula being a Boolean combination of
$\Sigma_\lang$-equation with variables in $X_\lang$ (i.e., atoms). The
terms in $T_{\ecal_\lang}$ are called \textit{built-ins} and represent
the portion of the specification that will be handled by the SMT
solver (i.e., semantic data types). Thus, an SMT instance is a formula
$\phi \in \oqff{\Sigma_\lang}{X_\lang}$ and the initial algebra
$\tcal_{\ecal_\lang^{+}}$, where $\ecal_\lang^{+}$ is a
\textit{decidable extension} of $\ecal_\lang$ (typically by adding
some inductive consequences and, perhaps, some extra symbols) such
that
\begin{align*}
	\phi
	\textnormal{ is satisfiable in $\tcal_{\ecal_\lang^+}$} \; \iff \; (\exists
	\func{\sigma}{X_\lang}{T_{\Sigma_\lang}})\; \tcal_{\ecal_\lang}\MODELS\phi\sigma.
\end{align*}
Many decidable theories $\ecal_{\lang}^+$ of interest are supported by
SMT solvers satisfying this requirement
(see~\cite{rocha-rewsmtjlamp-2017} for details).

The $\cde{INTEGER}$ module implements the equational theory
$\ecal_\lang = (\Sigma_\lang,E_\lang)$ of built-ins and the sort
$\cde{Boolean}$ defines the data type used to represent the
constraints. The topmost concurrent transitions in $\rcal$ are then
\textit{symbolic} rewrite steps of state terms with subterms in the
set $T_\Sigma(X_\lang)_\cde{Sys}$ of $\Sigma$-terms of sort
$\cde{Sys}$ with variables over the built-in sorts in $\Sigma_\lang$.

\begin{lemma}\label{lem.rew.constsys}
	The pair $\mathbf{B} =
	\left(\oqff{\Sigma_\lang}{X_\lang},\RMODELS\right)$ is
	a constraint system, where $X_\lang$ are the variables ranging over
	the sorts $\cde{Boolean}$ and $\cde{Integer}$.
\end{lemma}

\noindent The elements in $\oqff{\Sigma_\lang}{X_\lang}$ are
equivalence classes of quantifier-free $\Sigma_\lang$-formulas of sort
$\cde{Boolean}$ modulo semantic equivalence in $\tcal_{\ecal_\lang}$
(this technicality guarantees, e.g., the uniqueness of least upper
bounds). Therefore, by an abuse of notation, the constraint system
$\mathbf{B}$ has quantifier-free $\Sigma_\lang$-formulas of sort
$\cde{Boolean}$ as the constraints and the inverse $\RMODELS$ of the
semantic validity relation $\MODELS$, w.r.t. the initial model
$\tcal_{\ecal_\lang}$, as the entailment relation.

In order to use $\mathbf{B}$ as the underlying constraint system,
$\rcal$ relies on the current version of Maude that is integrated with
the CVC4~\cite{barrett-cvc4-2011} and
Yices2~\cite{dutertre-yices2-2014} SMT solvers. The $\cde{SMT-UTIL}$
module encapsulates this integration, which requires the reflective
capabilities of Maude available from the $\cde{META-LEVEL}$
module. The function $\cde{entails}$ implements the semantic validity
relation $\MODELS$ (w.r.t. $\tcal_{\ecal_\lang}$) using the auxiliary
functions $\cde{check-sat}$ and $\cde{check-unsat}$ (observe that the
sort $\cde{Boolean}$ is different to the usual sort $\cde{Bool}$ for
Boolean terms in Maude):

\begin{maude}
  op entails     : Boolean Boolean -> Bool .
  op check-sat   : Boolean -> Bool .
  op check-unsat : Boolean -> Bool .
  eq check-sat(B)
   = metaCheck(['INTEGER], upTerm(B)) .
  eq check-unsat(B)
   = not(check-sat(B)) .
  eq entails(C1, C2)
   = check-unsat(C1 and not(C2)) .
\end{maude}

\noindent The function invocation $\cde{check-sat(B)}$ evaluates to
true iff $\cde{B}$ is satisfiable; alternatively, it evaluates to
false if $\cde{B}$ is unsatisfiable or if the SMT solver times
out. The function invocation $\cde{check-unsat(B)}$ returns true iff
$\cde{B}$ is unsatisfiable. Note that if the constraints $\mathbf{B}$
are decidable, then $\cde{check-unsat}$ is not only sound but
complete. More precisely, if $\Gamma$ is a finite subset of decidable
constraints in $\mathbf{B}$ and $\phi$ is also a decidable constraint
in $\mathbf{B}$, then the following equivalence holds:
\[\tcal_{\ecal_\lang} \MODELS {\bigwedge_{\gamma\in\Gamma} \gamma}\Rightarrow \phi  \quad \iff \quad \mth{entails}\left({\bigwedge_{\gamma\in\Gamma} \gamma}\;, \phi\right).\]

\subsection{System States}

At the top level, the system is represented by the top sort
$\cde{Sys}$ defined in $\cde{SCCP-STATE}$:

\begin{maude}
  sort Sys .
  op {_} : Cnf -> Sys [ctor] .
\end{maude}

\noindent The argument of a state is the configuration of objects
representing the setup of the agents and processes in the system. Sort
$\cde{Cnf}$ is that of configuration of agents in an object-like
notation. More precisely, sort $\cde{Cnf}$ represents multisets of
terms of sort $\cde{Obj}$, with set union denoted by juxtaposition. An
object is by itself a configuration of objects, namely, the singleton
one; constant $\cde{mt}$ denotes the empty configuration and it is the
identity of the union operator. There are two types of objects:
process objects and store objects. Process and stores objects are
represented by triples $\cde{[\_,\_,\_]}$ (as $\cde{Obj}$). The first
two arguments of both a process and a store object are its type
(either $\cde{process}$ or $\cde{store}$, as $\cde{Cid}$) and its
identifier (as $\cde{Aid}$). The third argument of a process object is
the program it is executing (as $\cde{SCCPCmd}$) and the third
argument of a store object is a formula representing the constraint of
its corresponding agent (as $\cde{Boolean}$). 

\begin{maude}
  sorts Cid Obj Cnf .
  subsorts Obj < Cnf .
  ops store process : -> Cid .
  op [_,_,_] : Cid Aid Boolean -> Obj [ctor] .
  op [_,_,_] : Cid Aid SCCPCmd -> Obj [ctor] .
  op mt : -> Cnf [ctor] .
  op __ : Cnf Cnf -> Cnf [ctor assoc comm id: mt] .
  op {_} : Cnf -> Sys [ctor] .
\end{maude}

\noindent The idea is that in any observable state there can be many
process objects executing in an agent's space but there must be
exactly one store object per agent (i.e., space). More precisely, in
an observable state, each agent's space is represented by a set of
terms: some encoding the state of execution of all its processes and
exactly one object representing its local store.

Process and store objects use a qualified name (sort $\cde{Aid}$)
identifying to which agent's space they belong; this sort is defined
in module $\cde{AGENT-ID}$. Natural numbers (sort \cde{iNat}), in
Peano notation and with an equality
enrichment~\cite{gutierrez-eqenrscp-2015}, are used to specify agents'
identifiers. The hierarchical structure of spaces is modeled as a
tree-like structure where the root space is identified by the constant
$\cde{root}$. Any other qualified name corresponds to a dot-separated
list of agent identifiers, arranged from left to right. That is,
$\cde{3 . 1 . root}$ denotes that agent $\cde{3}$ is within the space
of agent $\cde{1}$, which in turn is within the top level of
$\cde{root}$.

\begin{maude}
  sorts iZero iNzNat iNat Aid .
  subsort iZero iNzNat < iNat .
  op root : -> Aid .
  op _._ : iNat Aid -> Aid .  
  op 0 : -> iZero [ctor] .
  op s_ : iNat -> iNzNat [ctor] .
  op _~_ : iNat iNat -> Bool [comm] .
\end{maude}

The processes in $\SCCPwE$ are modeled as \textit{commands} (sort
$\cde{SCCPCmd}$) as defined in the module $\cde{SCCP-SYNTAX}$:

\begin{maude}
  op 0 : -> SCCPCmd . 
  op tell_ : Boolean -> SCCPCmd . 
  op ask_->_ : Boolean SCCPCmd -> SCCPCmd .
  op _||_ : SCCPCmd SCCPCmd -> SCCPCmd [assoc comm gather (e E) ] .
  op <_>[_] : iNat SCCPCmd -> SCCPCmd .
  op rec(_,_) : iNat SCCPCmd -> SCCPCmd .
  op xtr(_,_) : iNat SCCPCmd -> SCCPCmd .
  op v(_) : iNat -> SCCPCmd .
\end{maude}

\noindent The argument of a $\cde{tell\_}$ command is a formula (as
$\cde{Boolean}$), namely, the formula to be added to the corresponding
store. The $\cde{ask\_->\_}$ command has a formula (as
$\cde{Boolean}$) and a program (as $\cde{SCCPCmd}$) as arguments,
denoting that if the given formula is entailed by the corresponding
store, then the process is to be executed next. Both arguments of the
$\cde{\_||\_}$ command are processes (as $\cde{SCCPCmd}$). The
arguments of the $\cde{<\_>[\_]}$, $\cde{rec(\_,\_)}$, and
$\cde{xtr(\_,\_)}$ commands are a natural number (representing the
identifier of a descendant, a variable, and the identifier of the ancestor, respectively) and a command to be executed. Note that
the syntax of each command is very close to the actual syntax in the
$\SCCPwE$ model, e.g., constructs of the form $P \| Q$ and
$\left[ P \right]_i$ in $\SCCPwE$ are represented in the syntax of
$\cde{SCCPCmd}$ by terms of the form $\cde{P || Q}$ and $\cde{<i>[P]}$,
respectively.

\begin{example}\label{exa.rew.state}
Using the functional module \cde{SCCP-STATE}, the $\SCCPwE$ space
structure in Example~\ref{fig.sccp.sccp-p} can be represented as
follows:

\begin{maude}
  { [store, root, true]
    [store, 0 . root, X:Integer === 25]
    [store, s 0 . root, true]
    [store, 0 . s 0 . root, Y:Integer < 5] }
\end{maude}
\end{example}

\subsection{Auxiliary Operations}

There are three auxiliary operations defined in the semantics. They
are used to verify whether a condition is satisfied by the system
before a transition happens and for replacing terms for the recursion
command.

The function symbol $\cde{is-prefix?}$ is defined in $\cde{AGENT-ID}$
to verify recursively whether an agent is descendant of other agent
based on its $\cde{Aid}$. For instance, $\cde{1 . 2}$ is descendant of
$\cde{2}$ because $\cde{2}$ is prefix of $\cde{1 . 2}$. By definition,
$\cde{root}$ is prefix of every agent. The equations defining the
$\cde{is-prefix?}$ function evaluate to $\cde{true}$ if the first
argument is a prefix of the second one; otherwise, it is
$\cde{false}$.

\begin{maude}
  op is-prefix? : Aid Aid -> Bool .
  eq is-prefix?(root, L) 
   = true . 
  eq is-prefix?(N . L, root) 
   = false .
  eq is-prefix?(N0 . L0, N1 . L1) 
   = (N0 . L0 ~ N1 . L1) or-else is-prefix?(N0 . L0, L1) .
\end{maude}

Some commands in $\SCCPwE$ require that a store exists in the
system. For that reason the $\cde{exists-store?}$ function symbol is
defined in module $\cde{SCCP-STATE}$: it is used to look within a
configuration of agents (as $\cde{Cnf}$) for a store with an given
$\cde{Aid}$. The equation defining $\cde{exists-store?}$ evaluate to
$\cde{true}$ if the store exists; otherwise, it is $\cde{false}$.

\begin{maude}  
  op exists-store? : Cnf Aid -> Bool .
  eq exists-store?(mt, L)
   = false .
  eq exists-store?( [process, L0, C0] X, L)
   = exists-store?(X,L) .
  eq exists-store?( [store, L0, B0] X, L)
   = (L0 ~ L) or-else exists-store?(X, L) .
\end{maude}

Finally, the recursion construct $P [ \mu x . P / x ]$ in $\SCCPwE$
means that every free occurrence of $x$ in $P$ is replaced with $\mu x
. P$. For this reason, the $\cde{replace}$ function symbol is defined as follows:
given a program $P$, a variable identifier $N$ (as $\cde{iNat}$), and
a program $C$, every free occurrence of $N$ in $P$ is substituted by $C$.

\begin{maude}  
  op replace : SCCPCmd iNat SCCPCmd -> SCCPCmd .
  eq replace( 0, N, C ) 
   = 0 . 
  eq replace( tell B, N, C ) 
   = tell B .
  eq replace( ask B -> C0, N, C ) 
   = ask B -> replace( C0, N, C ) ..
  eq replace( C0 || C1, N, C )
   = replace( C0, N, C ) || replace( C1, N, C ) .
  eq replace( < N0 >[ C0 ], N, C )
   = < N0 >[ replace( C0, N, C ) ] .
  eq replace( rec( N0, C0 ), N, C ) 
   = rec( N0, C0 ) .
  eq replace( xtr( N0, C0 ), N, C )  
   = xtr( N0, replace( C0, N, C ) ) .
  eq replace( v(N0), N, C ) 
   = if (N0 ~ N) then C else v(N0) fi .
\end{maude}

\subsection{System Transitions}

The state transitions in $\rcal$ comprise both invisible (given by
equations) and observable (given by rules) transitions.

There are two types of invisible transitions that are specified with
the help of equations. Namely, one to remove a $\cde{0}$ process from
a configuration and another one to join the contents of two stores of
the same space (i.e., two stores with the same $\cde{Aid}$). The
latter type of transition is important especially because when a new
process is spawned in a agent's space, a store with the empty
constraint (i.e., $\cde{true}$) is created for that space. If such
space existed before, then the idea is that the newly created store is
subsumed by the existing one (variable $\cde{L}$ is
of sort $\cde{Aid}$, $\cde{X}$ of sort $\cde{Cnf}$, and
$\cde{B0},\cde{B1}$ of sort $\cde{Boolean}$):

\begin{maude}
  eq { [ process, L0, 0 ] X }
   = { mt X } .
  eq [ store, L0, B0 ] [ store, L0, B1 ]
   = [ store, L0, B0 and B1 ] .
\end{maude}

There are six rules capturing the concurrent observable behavior in
the specification (variable $\cde{L}$ is of sort $\cde{Aid}$,
$\cde{X}$ of sort $\cde{Cnf}$, $\cde{B0},\cde{B1}$ of sort
$\cde{Boolean}$, $\cde{C0},\cde{C1}$ of sort $\cde{SCCPCmd}$, and
$\cde{N}$ of sort $\cde{iNat}$). 

\begin{maude}
  rl [tell] :
     { [ store, L0, B0 ] [process, L0, tell B1 ] X }
  => { [ store, L0, B0 and B1 ] [ process, L0, 0 ] X } .
  
 crl [ask] :
     { [ store, L0, B0 ] [ process, L0, ask B1  -> C1 ] X }
  => { [ store, L0, B0 ] [ process, L0, C1 ] X } 
  if entails(B0, B1) .
   
  rl [parallel] :
     { [ process, L0, C0 || C1 ] X }
  => { [ process, L0, C0 ] [ process, L0, C1 ] X } .
  
  rl [space] :
     { [ store, L0, B0 ] [ process, L0, < N0 >[ C0 ] ] X } 
  => { [ store, L0, B0 ] [ process, L0, 0 ] [ process, N0 . L0, C0] 
       [ store, N0 . L0, true ] X } .
  
  rl [recursion]:
     { [ process, L0, rec( N0, C0 ) ] X }
  => { [ process, L0, replace( C0, N0, rec( N0, C0 ) ) ] X } .
  
  rl [extrusion]:
     { [ process, N0 . L0, xtr( N0, C0 ) ] X }
  => { [ process, N0 . L0, 0 ] [ process, L0, C0 ] X } .
\end{maude}

\noindent Rule $\rlname{tell}$ implements the semantics of a process
executing a $\cde{tell}$ command by posting the given constraint in
the local store and by transforming such a process to the nil
process. Rule $\rlname{ask}$ executes command $\cde{C1}$ when the
guard $\cde{B1}$ in $\cde{ask B1 -> C1}$ holds: that is, when
$\cde{B1}$ is entailed by the local store $\cde{B0}$.  Rule
$\rlname{parallel}$ implements the semantics for parallel composition
of process by spawning the two process in the current space. Rule
$\rlname{space}$ creates a new space denoted by $\cde{N0 . L0}$ (as
$\cde{Aid}$) with an empty store (i.e., $\cde{true}$) and starts the
execution of program $\cde{C0}$ within the space. Rule
$\rlname{recursion}$ defines the semantics of a process executing a
$\cde{rec}$ command by using the aforementioned auxiliary function
$\cde{replace}$. Rule $\rlname{extrusion}$ executes process $\cde{C0}$
in the parent space of the agent and transitions the $\cde{xtr}$ process to the nil
process. Note that recursion command can lead to non-termination. It
is common in $\SCCPwE$ to guard such commands with an ask in order to
tame the potential non-termination.


\begin{example}\label{exa.rew.transition}
  Using the functional module \texttt{SCCP-SYNTAX}, the process $P$ in Figure~\ref{fig.sccp.sccp} can be represented as follows:
  
  \begin{maude}
  xtr(0,< s 0 >[tell (Z:Integer >= 10) || < 0 >[ask Y:Integer < 20 -> 
        xtr(0,xtr(s 0,< 0 >[< s s 0 >[tell (W:Integer < 
        Y:Integer)]]))]])
\end{maude}

  \noindent If this command is executed in the space of agent $\cde{0
    .  root}$ from the initial state in Example~\ref{exa.rew.state}
  (and depicted in Figure~\ref{fig.sccp.sccp-p}), it leads to
  the state

\begin{maude}
  { [store, root, true]
    [store, 0 . root, X:Integer === 25]
    [store, s 0 . root, Z:Integer >= 10]
    [store, s s 0 . 0 . root, W:Integer < Y:Integer]
    [store, 0 . s 0 . root, Y:Integer < 5] }
\end{maude}

\noindent which corresponds to the final state depicted in
Figure~\ref{fig.sccp.sccp-final}.
\end{example}

\begin{theorem}\label{thm.rew.main}
  If $\pairccp{P}{c}$ and $\pairccp{P'}{c'}$ are configurations of
  $\SCCPwE$ with underlying constraint system $\mathbf{B}$ restricted
  to decidable formulas, then
  \[\pairccp{P}{c} \stackrel{*}{\redi} \pairccp{P'}{c'} \quad \iff \quad \left\{ \overline{\pairccp{P}{c}}\right\} \stackrel{*}{\rews}_\rcal \left\{ \overline{\pairccp{P'}{c'}}\right\}, \]
  where $\overline{\pairccp{P}{c}}$ and $\overline{\pairccp{P'}{c'}}$
  are an encoding of the corresponding configurations in the syntax
  of $\rcal$.
\end{theorem}

\begin{proof}
  The proof follows by structural induction on the $\redi$ and
  $\rews_\rcal$ relations. \qed
\end{proof}
\section{Admissibility}
\label{sec.adm}

This section presents a map $\rcal \mapsto \rcal'$ that results in a
rewrite theory $\rcal'$, equivalent in terms of admissibility to
$\rcal$ (introduced in Section~\ref{sec.rew}) under some reasonable
assumptions, but in which dependencies to non-algebraic data types
such as terms over the built-ins or at the meta-level in $\rcal$ have
been removed. This section also presents proofs of admissibility of
the rewrite theory $\rcal'$. Such proofs are obtained mechanically
using the Maude Formal Environment
(MFE)~\cite{duran-mfecalco-2011,duran-mfetalcott-2011} and establish
the correspondence (i.e., soundness and completeness) between the
mathematical and operational semantics of $\rcal$. The specification
$\rcal'$ can be found in Appendix~\ref{sec.sccp.mfe}

The map $\rcal \mapsto \rcal'$ consists of the following items, which
make the specification amenable to mechanical verification in the
MFE:

\begin{itemize}
\item Changing the sort \cde{Bool} in \cde{TRUTH-VALUE} to the sort
  \cde{iBool} in the module \cde{ITRUTH-VALUE} and adjusting the
  specification to account for this new definition of Boolean values.
\item Removing all dependencies in \cde{SMT-UTIL} of the module
  \cde{META-LEVEL}.
\item Introducing a custom if-then-else-fi function symbol in
  \cde{SCCP-SYNTAX} and adjusting the definition of the auxiliary
  function symbol \cde{replace} to use this new version instead.
\end{itemize}

The module structure of the resulting specification $\rcal'$ is
depicted in Figure~\ref{fig.adm.rls-mfe}.  It is important to mention
that the dependency of \cde{META-LEVEL} in \cde{SMT-UTIL} amounts at
changing the definition of the function symbol \cde{check-sat} to a
constant value.
\begin{figure}[thpb] 
  \centering
  \begin{tikzpicture}[node distance=1.35cm,
    every node/.style={fill=white, font=\sffamily}, align=center]
    \node (SCCP)    [module]                
    			    {SCCP};
    \node (STATE)   [module, below of=SCCP]        
    			    {SCCP-STATE};
    \node (SYNTAX)  [module, below of=STATE]  
    			    {SCCP-SYNTAX};
    \node (AGENT)   [module, below of=SYNTAX, xshift=-2cm] 
    			    {AGENT-ID};
    \node (INAT)    [module, below of=AGENT, xshift=-2cm]    
    			    {INAT};
    \node (EXT)     [module, below of=AGENT, xshift=2cm]
    			    {IEXT-BOOL};
    \node (BOOL)    [module, below of=EXT]
    			    {IBOOL};
    \node (OPS)     [module, below of=BOOL] 
    			    {IBOOL-OPS};
    \node (TRUTH) 	[module, below of=OPS]
    			    {ITRUTH-VALUE};
    \node (SMT)     [module, below of=SCCP, right of=EXT, xshift=2.5cm]
    			    {SMT-UTIL};
    \node (INTEGER) [module, below of=SMT, right of=BOOL, xshift=2.5cm]
    			    {INTEGER};
    \node (BOOLEAN) [module, below of=INTEGER, right of=OPS, xshift=2.5cm]
    			    {BOOLEAN};
    \draw[triplearrow]   (SYNTAX) -- (STATE);
    \draw[triplearrow]    (AGENT) |- (SYNTAX);
    \draw[triplearrow]     (INAT) |- (AGENT);
    \draw[triplearrow]      (EXT) |- (AGENT);
    \draw[triplearrow]     (BOOL) -- (EXT);
    \draw[triplearrow]      (OPS) -- (BOOL);
    \draw[triplearrow]      (OPS) -- +(1.7,0) |- (SMT);
    \draw[triplearrow]    (TRUTH) -- (OPS);
    \draw[triplearrow]    (TRUTH) -| (INAT);
    \draw[triplearrow]  (BOOLEAN) -- (INTEGER);
    \draw[triplearrow]  (INTEGER) -- +(-1.7,0) |- (SYNTAX);
    \draw[triplearrow]      (SMT) |- (SCCP);    
    \draw[->]     (STATE) -- (SCCP);
    \draw[->]   (INTEGER) -- (SMT);
  \end{tikzpicture}
  \caption{Module structure of $\rcal'$.}
  \label{fig.adm.rls-mfe}
\end{figure}

\subsection{Equational Admissibility}

Recall from Section~\ref{sec.prelim} that an equational specification
is admissible if it is sort-decreasing, operational terminating, and
confluent.

\begin{lemma}\label{lem.adm.eqadm}
  The equational subtheory of $\rcal'$ is admissible.
\end{lemma}

\begin{proof}
  Termination of $\rcal'$ modulo axioms can be proved automatically
  using the Maude Termination Tool within the MFE with the following
  script:
\begin{lstlisting}[basicstyle=\ttfamily\lst@ifdisplaystyle\scriptsize\fi]
  Maude> (select tool MTT .)
  rewrites: 80 in 4ms cpu (4ms real) (20000 rewrites/second)
  The MTT has been set as current tool.
  
  Maude> (select external tool aprove .)
  rewrites: 39 in 4ms cpu (5ms real) (9750 rewrites/second)
  aprove is now the current external tool.
  
  Maude> (ctf SCCP .)
  rewrites: 839665 in 187688ms cpu (191102ms real) (4473 rewrites/second)
  Success: The functional part of module SCCP is terminating.
\end{lstlisting}
Sort-decreasingness and confluence modulo axioms can be proved
automatically using the Church-Rosser Checker within MFE with the
following script:

\begin{lstlisting}[basicstyle=\ttfamily\lst@ifdisplaystyle\scriptsize\fi]
  Maude> (select tool CRC .)
  rewrites: 76 in 4ms cpu (5ms real) (19000 rewrites/second)
  The CRC has been set as current tool.
  
  Maude> (ccr SCCP .)
  rewrites: 35929405 in 45540ms cpu (45539ms real) (788963 rewrites/second)
  Church-Rosser check for SCCP
    All critical pairs have been joined.
    The specification is locally-confluent.
    The module is sort-decreasing.
\end{lstlisting} \qed
\end{proof}

\subsection{Coherence}

Coherence follows from the fact that there are not critical pairs
between the equations and the rules in $\rcal'$: the key observation
is that $\rcal'$, similar to $\rcal$, is topmost.

\begin{lemma}\label{lem.adm.coh}
  The rewrite theory $\rcal'$ is coherent modulo axioms.
\end{lemma}

\begin{proof}
  Coherence of $\rcal'$ can be proved automatically using the Maude
  Coherence Checker within the MFE with the following script:

\begin{lstlisting}[basicstyle=\ttfamily\lst@ifdisplaystyle\scriptsize\fi]
  Maude> (select tool ChC .) 
  rewrites: 76 in 16ms cpu (14ms real) (4750 rewrites/second)
  The ChC has been set as current tool.
  
  Maude> (cch SCCP .)
  rewrites: 8004447 in 7844ms cpu (7843ms real) (1020454 rewrites/second)
  Coherence checking of SCCP
    All critical pairs have been rewritten and no rewrite with rules can
    happen at non-overlapping positions of equations left-hand sides.
\end{lstlisting} \qed
\end{proof}

\subsection{Admissibility}

The admissibility of $\rcal'$ is a logical consequence of its
equational admissibility and coherence.

\begin{theorem}\label{thm.adm.adm}
  The rewrite theory $\rcal'$ is admissible.
\end{theorem}

\begin{proof}
  It follows from lemmas~\ref{lem.adm.eqadm} and~\ref{lem.adm.coh}.  \qed
\end{proof}

Finally, the admissibility of $\rcal$ can be asserted under the
assumption that Maude's \cde{META-LEVEL} is admissible. In particular,
it is required that the meta-level functionality used for querying the
SMT solver is correct.

\begin{corollary}\label{cor.adm.adm}
  If Maude's \cde{META-LEVEL} is admissible, then $\rcal$ is admissible.
\end{corollary}

\begin{proof}[Sketch]
  The equational admissibility of $\rcal$ follows from the
  admissibility of $\rcal'$ by observing that:
  \begin{itemize}
    \item the equational subtheory of $\rcal'$ is basically the same
      one of $\rcal$ but with some sorts being renamed, introducing a
      new if-then-else-fi construct (with the same evaluation strategy
      that Maude's built-in equivalent).
    \item by assumption, Maude's \cde{META-LEVEL} is admissible.
  \end{itemize}
  The coherence of $\rcal$ follows from the coherence of $\rcal'$ and
  by the fact that there are no critical pairs between the equations
  in Maude's \cde{META-LEVEL} module and the rules in the \cde{SCCP}
  module. \qed
\end{proof}
\section{Symbolic Reachability Analysis}
\label{sec.reach}

The goal of this section is to explain how the rewriting logic
semantics $\rcal$ of $\SCCPwE$ and rewriting modulo SMT can be used as
an automatic mechanism for solving existential reachability goals in
the initial model $\tcal_\rcal$. This approach can be especially
useful for symbolically proving or disproving safety properties of
$\tcal_\rcal$ such as fault-tolerance, consistency, and privacy.  The
approach presented in this section mainly relies on Maude's
\cde{search} command, but it can be easily extended to be useful in
the more general setting of Maude's LTL Model Checker.

In this section, a state in $\rcal$ with $n$ stores is represented as
a term $t(\phi_1,\ldots,\phi_n)$, where each $\phi_i$ denotes the
contents of a store. Given two state terms $t(\phi_1,\ldots,\phi_n)$
and $u(\psi_1,\ldots,\psi_k)$, the existential reachability question
of whether there is a ground substitution $\theta$ and concrete states
$t'\in t(\phi_1,\ldots,\phi_n)\theta$ and $u'\in
u(\psi_1,\ldots,\psi_k)\theta$ such that $t'\stackrel{*}{\rews}_\rcal
u'$ is of special interest for many safety properties. For example,
$u'$ can represent a `bad state' and the goal is to know if reaching
such a state is possible.



\subsection{Fault-tolerance and Consistency}

Fault tolerance is the property that ensures a system to continue
operating properly in the event of the failure; consistency means that
a local failure does not propagate to the entire system. In $\rcal$,
this means that if a store becomes inconsistent, it is not the case
that such an inconsistency spreads to the entire system. Of course,
inconsistencies can appear in other stores due to some unrelated
reasons.

Finding an inconsistent store can be logically formulated by the
following model-theoretic satisfaction:
\[\tcal_\rcal \MODELS (\exists \overrightarrow{x}, i\in[1..k])\; t(\phi_1,\ldots,\phi_n) \stackrel{*}{\rews}_\rcal u(\psi_1,\ldots,\psi_k) \;\land\; \mth{unsat}(\psi_i).\]
Answering this query in the positive would mean that from some initial
state satisfying the pattern $t(\phi_1,\ldots,\phi_n)$, there is a
state in which a store becomes inconsistent.

Such queries can be easily implemented with the help of $\rcal$ and
the rewriting modulo SMT approach by using Maude's \cde{search}
command.  As an example, consider the following \cde{search} command:

  \begin{maude}
search in SCCP :
    { [store,   root, true] 
      [store, 0 . root, X:Integer === 25]
      [store, s 0 . root, true]
      [store, 0 . s 0 . root, Y:Integer < 5] 
      [process, 0 . root, xtr(0,< s 0 >[tell (Z:Integer >= 10) || 
               < 0 >[ask Y:Integer < 20 -> xtr(0,xtr(s 0,
               < 0 >[< s s 0 >[tell (W:Integer < Y:Integer)]]))]])] }
=>* { [store, A:Aid, B:Boolean, B0:Boolean] C:Cnf } 
    such that check-unsat(B0:Boolean) .
  \end{maude}

\noindent Note that a store is inconsistent if it is unsatisfiable,
thereby checking whether a store is inconsistent is accomplished with
the function \cde{check-unsat}. The aforementioned command generates
the following output:
  
\begin{maude}
No solution.
states: 19  rewrites: 1103 in 104ms cpu (101ms real) 
        (10605 rewrites/second)
\end{maude}

\noindent The command does not find an inconsistent store in the 19
reachable states. However, it is possible to make a store inconsistent
by adding inconsistent information, for example \cde{$Z \geq 10$} and
\cde{$Z = 9$} by changing the process \cde{tell(Z >= 10)} to
\cde{tell(Z >= 10) || tell(Z === 9)}:
  
\begin{maude}
Solution 1 (state 14)
states: 15  rewrites: 678 in 76ms cpu (76ms real)
        (8921 rewrites/second)
C:Cnf --> [process,s 0 . root,< 0 >[ask Y:Integer < 20 -> 
           xtr(0,xtr(s 0,< 0 >[< s s 0 >[tell (W:Integer < Y:Integer)]
           ]))]] 
          [store,root,true] 
          [store,0 . root,X:Integer === (25).Integer] 
          [store,0 . s 0 . root,Y:Integer < 5]
A:Aid --> s 0 . root
B0 --> Z:Integer >= (10).Integer and Z:Integer === (9).Integer

...

Solution 16 (state 54)
states: 55  rewrites: 3160 in 300ms cpu (299ms real) 
        (10533 rewrites/second)
C:Cnf --> [store,root,true] 
          [store,0 . root,X:Integer === (25).Integer] 
          [store,0 . s 0 . root,Y:Integer < 5] 
          [store,s s 0 . 0 . root,W:Integer < Y:Integer]
A:Aid --> s 0 . root
B0 --> Z:Integer === (9).Integer and Z:Integer >= (10).Integer

No more solutions.
states: 55  rewrites: 3169 in 300ms cpu (301ms real) 
        (10563 rewrites/second)
\end{maude}

\noindent There are 55 reachable states (from the initial state) and
16 of them have an inconsistent store. Note that, even though the
inconsistency appears for the first time in state 14, the system
evolves until no more processes can be performed. It is possible to
verify that there are states with consistent and inconsistent stores
at the same time by slightly modifying the command.

\subsection{Knowledge Inference}

Knowledge inference refers to acquiring new knowledge from existing
facts. In the setting of $\rcal$, this means that from a given initial
state an agent, at some point, has gained enough information to infer
--~from the rules of first-order logic~-- new facts:
\[\tcal_\rcal \MODELS (\exists \overrightarrow{x}, i\in[1..k])\; t(\phi_1,\ldots,\phi_n) \stackrel{*}{\rews}_\rcal u(\psi_1,\ldots,\psi_k) \;\land\; \psi_i \Rightarrow \tau,\]
where $\tau$ is the formula representing the new fact and
`$\Rightarrow$' denotes logical implication. Answering the above query
in the positive, means that from some initial state satisfying the
pattern $t(\phi_1,\ldots,\phi_n)$ there is at least an agent as part
of a configuration satisfying the pattern $u(\psi_1,\ldots,\psi_k)$
that has enough information to infer $\tau$. 

Such queries can be easily implemented with the help of $\rcal$ and
the rewriting modulo SMT approach by using Maude's \cde{search}
command.  As an example, consider the following \cde{search} command:

\begin{maude}
search in SCCP :
    { [store,   root, true] 
      [store, 0 . root, X:Integer === 25]
      [store, s 0 . root, true]
      [store, 0 . s 0 . root, Y:Integer < 5] 
      [process, 0 . root, xtr(0,< s 0 >[tell (Z:Integer >= 10) || 
               < 0 >[ask Y:Integer < 20 -> xtr(0,xtr(s 0,
               < 0 >[< s s 0 >[tell (W:Integer < Y:Integer)]]))]])] }  
    =>*
    { [store, A:Aid, B0:Boolean] C:Cnf }
    such that entails(B0:Boolean, Y:Integer > 9) .
\end{maude}

\noindent It checks if there is a state, reachable from the given
initial state, in which some store logically implies $Y > 9$.  The
output below shows that such a state does not exist:

\begin{maude}
No solution.
states: 19  rewrites: 1535 in 116ms cpu (116ms real) 
        (13232 rewrites/second)
\end{maude}

\noindent However, if the condition in the command is changed to find
a store logically implying $Z > 9$, the query finds 8 solutions:

\begin{maude}
Solution 1 (state 4)
states: 5  rewrites: 318 in 28ms cpu (27ms real) 
        (11357 rewrites/second)
C:Cnf --> [process,s 0 . root,< 0 >[ask Y:Integer < 20 -> 
           xtr(0,xtr(s 0,< 0 >[< s s 0 >[tell (W:Integer < Y:Integer)]
           ]))]] 
          [store,root,true] 
          [store,0 . root,X:Integer === (25).Integer] 
          [store,0 . s 0 . root,Y:Integer < 5]
A:Aid --> s 0 . root
B0:Boolean --> Z:Integer >= (10).Integer

...

Solution 8 (state 18)
states: 19  rewrites: 1520 in 120ms cpu (120ms real) 
        (12666 rewrites/second)
C:Cnf --> [store,root,true] 
          [store,0 . root,X:Integer === (25).Integer] 
          [store,0 . s 0 . root,Y:Integer < 5] 
          [store,s s 0 . 0 . root,W:Integer < Y:Integer]
A:Aid --> s 0 . root
B0:Boolean --> Z:Integer >= (10).Integer

No more solutions.
states: 19  rewrites: 1535 in 124ms cpu (121ms real) 
        (12379 rewrites/second)
\end{maude}

\subsection{Same Knowledge}

Formally, this type of reachability query can be specified as follows for $\rcal$:
\[\tcal_\rcal \MODELS (\exists \overrightarrow{x}, i,j\in[1..k])\; t(\phi_1,\ldots,\phi_n) \stackrel{*}{\rews}_\rcal u(\psi_1,\ldots,\psi_k) \;\land\; \psi_i \Leftrightarrow \psi_j \;\land\; i\neq j ,\]
where $\Leftrightarrow$ denotes logical equivalence. Answering this
query in the positive means that two stores have gained the same
knowledge.

As an example, consider the following Maude \cde{search} command, querying
for two stores having the same information when they are non-empty:

  \begin{maude}
search in SCCP :
    { [store,   root, true] 
      [store, 0 . root, X:Integer === 25]
      [store, s 0 . root, true]
      [store, 0 . s 0 . root, Y:Integer < 5] 
      [process, 0 . root, xtr(0,< s 0 >[tell (Z:Integer >= 10) || 
               < 0 >[ask Y:Integer < 20 -> xtr(0,xtr(s 0,
               < 0 >[< s s 0 >[tell (W:Integer < Y:Integer)]]))]])] }
=>* { [store, A0:Aid, C0:Boolean] 
      [store, A1:Aid, C1:Boolean] C:Cnf } 
    such that entails(C0:Boolean, C1:Boolean) /\ 
              entails(C1:Boolean, C0:Boolean) /\ 
              C1:Boolean =/= true .  
  \end{maude}

\noindent Using Figure~\ref{fig.sccp.sccp-final} is easy to check that
it is never the case that there are two stores with the same
information, which agrees with the output shown below.

\begin{maude}
No solution.
states: 19  rewrites: 3899 in 472ms cpu (471ms real) 
        (8260 rewrites/second)  	
\end{maude}
  
Replacing the process \cde{tell(W < Y)} with \cde{tell(Z > 9)}
leads to a state where two stores have the same information, namely,
\cde{tell(Z > 9)} and \cde{tell(Z >= 10)}. The output for the
\cde{search} command in this case is the following:
  
\begin{maude}
Solution 1 (state 18)
states: 19  rewrites: 3862 in 464ms cpu (462ms real) 
        (8323 rewrites/second)
C:Cnf --> [store,root,true] 
          [store,0 . root,X:Integer === (25).Integer] 
          [store,0 . s 0 . root,Y:Integer < 5]
A0:Aid --> s 0 . root
C0:Boolean --> Z:Integer >= (10).Integer
A1:Aid --> s s 0 . 0 . root
C1:Boolean --> Z:Integer > (9).Integer

Solution 2 (state 18)
states: 19  rewrites: 3914 in 472ms cpu (471ms real) 
        (8292 rewrites/second)
C:Cnf --> [store,root,true] 
          [store,0 . root,X:Integer === (25).Integer] 
          [store,0 . s 0 . root,Y:Integer < 5]
A0:Aid --> s s 0 . 0 . root
C0:Boolean --> Z:Integer > (9).Integer
A1:Aid --> s 0 . root
C1:Boolean --> Z:Integer >= (10).Integer

No more solutions.
states: 19  rewrites: 3917 in 472ms cpu (471ms real) 
        (8298 rewrites/second)
  \end{maude}
\section{An $\SCCPwE$-based Programming Environment Prototype}
\label{sec.lang}

This section present a programming environment prototype based on
$\SCCPwE$ and whose executable semantics is given by $\rcal$. The goal
of this prototype is to provide programmers with the theoretical
fundamentals and expressiveness of $\SCCPwE$, an easy syntax, and a
front-end to interact with $\rcal$.

\subsection{Syntax of the Programming Language}
 
The syntax of the programming language is presented using the EBNF
notation in Figure~\ref{fig.lang.ebnf}. A program $\langle system
\rangle$ has two sections: the header $\langle variables \rangle$ and
the body $\langle body \rangle$. The former contains the variable
declarations by a name $\langle id \rangle$ and a type (viz.,
\cde{Bool} and \cde{int}), it is not possible to declare two variables
with the same name and different type. The latter contains an unsorted
list of agents and processes, one per line, between the keywords
\lit{begin} and \lit{end}. Note that each line describing an agent or
a process shall end with the character \lit{.}.  Since the purpose is
to provide an easier way to write $\SCCPwE$ systems, it is possible to
declare Boolean (as \cde{Bool}) and integer (as \cde{Int}) variables
in order to make agents and processes expressions simpler.

\begin{figure}[thbp] 
  \centering
  \setlength{\grammarindent}{10em} 
  \begin{grammar}
   <system> ::= <variables>* <body> 
    
   <variables> ::= \lit{var} <idList>+ (\lit{Int} | \lit{Bool})
    
   <idList> ::= <id> (\lit{,} <id>)*
   
   <body> ::= \lit{begin} <line>+ \lit{end}
   
   <line> ::= ( <agent> | <process> ) \lit{.}
   
   <agent> ::= <location> \lit{;} <constraint>
   
   <process> ::= \lit{tell(} <constraint>  \lit{)}  \hfill (tell)
   \alt \lit{ask} <constraint> \lit{->}  <process> \hfill (ask)
   \alt <process> \lit{||} <process>				 \hfill (parallel)
   \alt \lit{[} <process> \lit{]_} <integer>		 \hfill (space)
   \alt \lit{x(} <process> \lit{)_} <integer>		 \hfill (extrusion)
   \alt \lit{v(} <integer> \lit{)}					 \hfill (variable)
   \alt \lit{r(} <integer> \lit{,} <process> \lit{)} \hfill (recursion)
   
   <constraint> ::= <boolean>
   \alt <id>
   \alt <expression>
   \alt <constraint> \lit{and} <constraint>
   
   <location> ::= (<integer> \lit{.})* \lit{root}
    
   <expression> ::= <id> <operator> ( <id> | <integer> )
   
   <operator> ::= `>' | `<' | `=' | `=/=' | `>=' | `<=' 
   
   <boolean> ::= \lit{true} | \lit{false}
   
   <integer> ::= [0-9]+
   
   <id> ::= [A-Z] [A-Z0-9]*
  \end{grammar}
  \caption{Syntax of the programming language prototype.}
  \label{fig.lang.ebnf}
\end{figure}

\subsection{Examples}

The Example~\ref{sec.sccp.example} can be represented in the programming
language as follows:

\begin{sccp}
var W, X, Y, Z Int
begin
root ; true .
0 . root ; X = 25 .
1 . root ; true .
0 . 1 . root ; Y < 5 .
[ x( [tell(Z >= 10) || [ask Y < 20 -> x( x( [ [tell(W < Y)]_2 ]_0 
)_1 )_0 ]_0 ]_1 )_0 ]_0 .
end
\end{sccp}

\noindent Note that every variable used in a constraint is
\textit{declared} or \textit{defined} before the keyword \cde{begin},
regardless of its type (\cde{Bool} or \cde{Int}). Since $\SCCPwE$
configurations are represented in Maude as a \textit{soup} of agents
and processes, the order in which these are defined does not matter.
However, it is a good idea to define agents first for the ease of
reading and self-documentation.

\begin{example}

Consider the following program,

\begin{sccp}
var B0, B1 Bool
var X, C Int
var Y, B Int
begin
ask true -> tell(X >= 5) .
[ [tell(B0)]_1 || tell(Y < X) ]_1 .
ask X > 1 -> tell(B1) .
[tell(X >= 5)]_2 .
ask B1 -> [ [tell(C =/= 5) ]_1 ]_1 .
[ask Y < 3 -> r(1,v(1) || tell(false)) ]_1 .
end
\end{sccp}

\begin{figure}[htbp] 
	\centering
	\resizebox{6cm}{!}{%
		\begin{tikzpicture}[node distance=2cm,
		every node/.style={fill=white, font=\sffamily}, align=center]
		\node (root)  	[store]                
						{\textit{(root)}\\X >= 5 \\ $\sqcup$ B1};
		\node (1)		[store, below of=root, left of=root]
						{\textit{(1)}\\Y < X};
		\node (2)  		[store, below of=root, right of=root]  
						{\textit{(2)}\\X >= 5};
		\node (11)   	[store, below of=1] 
						{\textit{(1)}\\C =/= 5 \\ $\sqcup$ B0};
		\draw[->]   (root) -- (1);
		\draw[->]   (root) -- (2);
		\draw[->]      (1) -- (11);      
		\end{tikzpicture}}
	\caption{$\SCCPwE$ System example.}
	\label{fig.lang.state}
\end{figure}
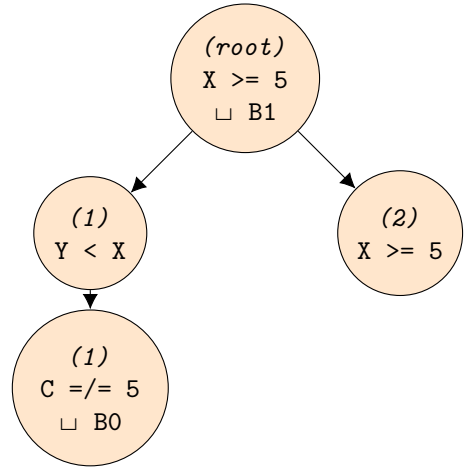

This program shall be reduced to the final state showed in
Figure~\ref{fig.lang.state}. There are some points to recall:

\begin{itemize}
  \item it is not necessary to include agents in the description of the 
  system, because these can come of processes,
  \item the process $\cde{[ask Y < 3 -> r(1,v(1) || tell(false)) ]_1}$ 
  shall not be reduced due to the fact that $\cde{Y < 3}$ is not 
  entailed by $\cde{Y < X}$, i.e.,  
  $\cde{Y < 3} \not\sqsubseteq \cde{Y < X}$,
  \item the space of an agent can be represented by a boolean variable 
  ($\cde{Bool}$).
\end{itemize}

The provided programming environment is a graphical tool where a
programmer can execute $\SCCPwE$ programs using the aforementioned
programming language. This tool is developed with Python 3 and
tkinter. Figure~\ref{fig.lang.tool} shows an example of the tool's GUI:
it has a main window~\ref{fig.lang.tool-initial} where programs can be
written. Once a valid program is written and executed, the final state
of the $\SCCPwE$ system is shown on an auxiliary
window~\ref{fig.lang.tool-final}.

\begin{figure}[htbp] 
	\centering
	\begin{subfigure}[b]{0.83\textwidth}
		\includegraphics[width=4in]{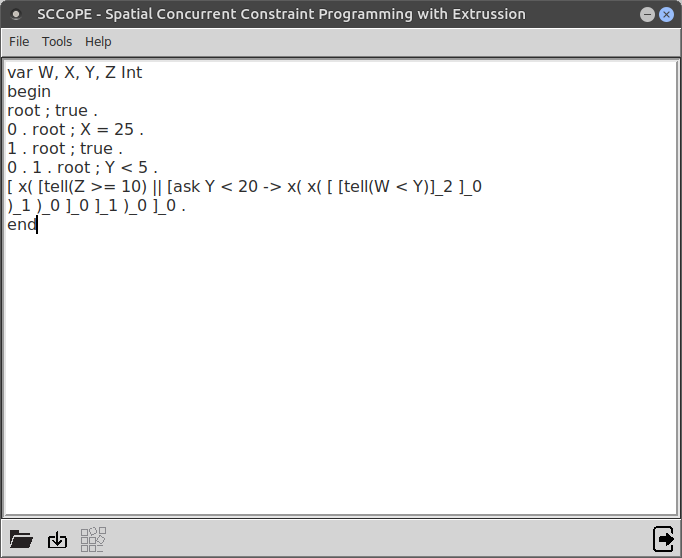} 
		\caption{tool's GUI and $\SCCPwE$ program}
		\label{fig.lang.tool-initial}
	\end{subfigure}%
	
	\begin{subfigure}[b]{0.83\textwidth}
		\includegraphics[width=4in]{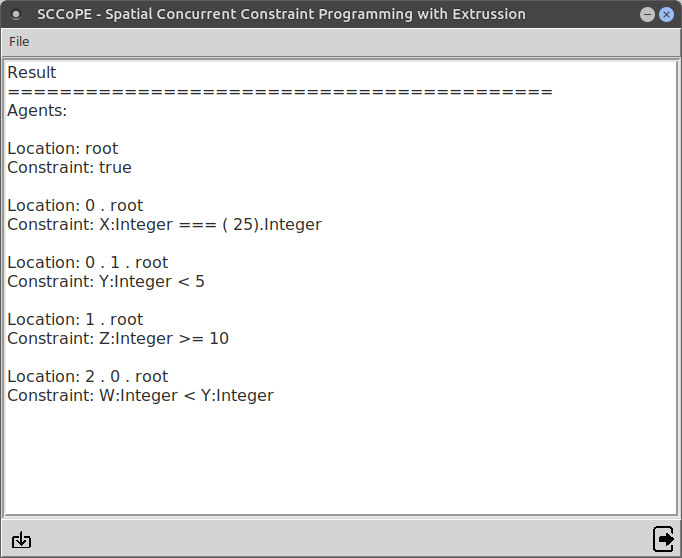} 
		\caption{final state and output format}
		\label{fig.lang.tool-final}
	\end{subfigure}	
	\caption{Programming environment for the $\SCCPwE$-based programming
    language.}
	\label{fig.lang.tool}
\end{figure}
	
\end{example}
\section{Related Work and Concluding Remarks}
\label{sec.concl}

Rewrite-based executable semantics of process-based formalisms have
been proposed before in the realm of rewriting logic and Maude (see,
e.g.,~\cite{degano-ccs-2002,braga-modularsem-2005,verdejo-twocasessem-2005}).
They are part of a larger set of formal interpreters developed over
the years that have helped in exploring the features of rewriting
logic as a semantic framework.  The work presented here is a
significant extension of the preliminary work initiated
in~\cite{romero-bsc-2017}. In particular, this work adds support for
the recursion and extrusion primitives present in $\SCCPwE$. The
related work in~\cite{barco-kstores-2012}, presents an interpreter for
epistemic and spatial modalities in Prolog.

This paper has presented a symbolic rewriting logic semantics --~based
on the rewriting modulo SMT approach~-- of
$\SCCPwE$~\cite{knight-sccp-2012,barco-kstores-2012,guzman-sccp-2016}:
a recent extension of the $\CCP$
model~\cite{saraswat-ccp-1990,saraswat-ccpbook-1993,saraswat-ccpsem-1991}
with spaces and extrusion. The executable rewriting logic semantics
follows the structural operational semantics of $\SCCPwE$ and
implements the underlying constraint system using SMT-solving
technology. As such, it offers a complete and sound decision procedure
for symbolic reachability analysis in $\SCCPwE$ for existential
formulas, that can be automatically mechanized in Maude. Several
examples have been used to illustrate the main concepts and a
programming environment prototype has been introduced. The novel idea
of combining term rewriting and constrained data structures, as it is
the case in the rewriting modulo SMT
approach~\cite{rocha-rewsmtjlamp-2017}, is an active area of
research. Ultimately, this approach strengthens with symbolic support
the battery of techniques that can now be used to implement formal and
symbolic executable semantics of languages in Maude.

As described in~\cite{olarte-emergmodels-2013} there are several
extensions and applications of the $\CCP$ model, e.g., the epistemic
and spatial modalities, mobile behavior, linear and soft modalities,
probabilistic behavior, and timed concurrent constraint
programming. As future work, extensions of $\SCCPwE$ with
probabilities and time are a promising line of research. Moreover,
providing the rewriting logic semantics of such extensions can lead to
interesting case studies for PMaude~\cite{agha-pmaude-2006} and
Real-Time Maude~\cite{olveczky-realtimemaude-2000}. Finally, new case
studies with applications to emergent systems such as cloud computing
and social networks should be pursued with the help of the rewriting
logic semantics presented in this work.

\bibliographystyle{abbrv}
\bibliography{camilo,miguel}

\newpage

\section{SCCP specification}
\label{sec.sccp.maude}

This appendix includes the $\SCCP$ specification in Maude explained in
Section~\ref{sec.rew}, including the functional modules $\cde{INAT}$,
$\cde{SMT-UTIL}$, $\cde{AGENT-ID}$, $\cde{SCCP-SYNTAX}$ and
$\cde{SCCP-STATE}$, and the system module $\cde{SCCP}$.

\subsection{prelude-short.maude}

\begin{maude}
---- Natural numbers
fmod INAT is
  protecting TRUTH-VALUE .
  sorts iZero iNzNat iNat .
  subsort iZero iNzNat < iNat .
  
  op 0 : -> iZero [ctor] .
  op s_ : iNat -> iNzNat [ctor] .
  
  --- equality enrichment
  op _~_ : iNat iNat -> Bool [comm] .
  eq 0 ~ 0 
   = true .
  eq s N:iNat ~ 0 
   = false .
  eq N:iNat ~ N:iNat
   = true .
  eq s N:iNat ~ s M:iNat
   = N:iNat ~ M:iNat .
endfm
\end{maude}

\subsection{smt-util.maude}

\begin{maude}
load smt.maude

fmod SMT-UTIL is
  inc INTEGER .
  pr CONVERSION .
  pr META-LEVEL .
  
  op check-sat : Boolean -> Bool .
  op check-unsat : Boolean -> Bool .
  op entails : Boolean Boolean -> Bool .
  eq check-sat(B:Boolean)
   = metaCheck(['INTEGER], upTerm(B:Boolean)) .
  eq check-unsat(B:Boolean)
   = not(check-sat(B:Boolean)) .
  eq entails(C1:Boolean, C2:Boolean)
   = check-unsat(C1:Boolean and not(C2:Boolean)) .

  --- some Boolean identities
  eq B:Boolean and true
   = B:Boolean .
  eq B:Boolean and false
   = false .
  eq B:Boolean or true
   = true .
  eq B:Boolean or false
   = B:Boolean .
  eq true and B:Boolean
   = B:Boolean .
  eq false and B:Boolean
   = false .
  eq true or B:Boolean
   = true .
  eq false or B:Boolean
   = B:Boolean .
  eq not((true).Boolean)
   = (false).Boolean .
  eq not((false).Boolean)
   = (true).Boolean .
endfm
\end{maude}

\subsection{agent-id.maude}

\begin{maude}
load smt-util.maude
load prelude-short.maude

--- agent identifier
fmod AGENT-ID is
  pr EXT-BOOL .
  pr INAT .
  sort Aid .
  
  op root : -> Aid .
  op _._ : iNat Aid -> Aid .
  
  vars L L0 L1 L2 : Aid .
  vars N N0 N1 N2 : iNat .
  
  --- auxiliary operations
  op is-prefix? : Aid Aid -> Bool .
  eq is-prefix?(root, L) 
   = true . 
  eq is-prefix?(N . L, root) 
   = false .
  eq is-prefix?(N0 . L0, N1 . L1) 
   = (N0 . L0  ~ N1 . L1) or-else is-prefix?(N0 . L0, L1) .
  
  --- equality enrichment
  op _~_ : Aid Aid -> Bool [comm] .
  eq root ~ root
   = true .
  eq root ~ N . L
   = false . 
  eq L ~ L
   = true .
  eq N . L ~ N0 . L0
   = (N ~ N0) and-then L ~ L0 .
endfm
\end{maude}

\subsection{sccp.maude}

\begin{maude}
load agent-id.maude

--- commands syntax
fmod SCCP-SYNTAX is
  pr INTEGER .
  pr AGENT-ID .
  sort SCCPCmd .
  
  op 0 : -> SCCPCmd . 
  op tell_ : Boolean -> SCCPCmd . 
  op ask_->_ : Boolean SCCPCmd -> SCCPCmd .
  op _||_ : SCCPCmd SCCPCmd -> SCCPCmd [assoc comm gather (e E) ] .
  op <_>[_] : iNat SCCPCmd -> SCCPCmd .
  op rec(_,_) : iNat SCCPCmd -> SCCPCmd .
  op xtr(_,_) : iNat SCCPCmd -> SCCPCmd .
  op v(_) : iNat -> SCCPCmd .
endfm

--- state syntax
fmod SCCP-STATE is
  pr SCCP-SYNTAX .
  
  sorts Cid Obj Cnf Sys .
  subsorts Obj < Cnf .
  ops store process : -> Cid .
  op [_,_,_] : Cid Aid Boolean -> Obj [ctor] .
  op [_,_,_] : Cid Aid SCCPCmd -> Obj [ctor] .
  op mt : -> Cnf [ctor] .
  op __ : Cnf Cnf -> Cnf [ctor assoc comm id: mt] .
  op {_} : Cnf -> Sys [ctor] .
  
  vars L L0 L1 : Aid .
  vars N N0 N1 : iNat .
  vars B B0 B1 : Boolean .
  vars C C0 C1 : SCCPCmd .
  vars X Y     : Cnf .

  --- auxiliary operations
  op replace : SCCPCmd iNat SCCPCmd -> SCCPCmd .
  eq replace( 0, N, C ) 
   = 0 . 
  eq replace( tell B, N, C ) 
   = tell B .
  eq replace( ask B -> C0, N, C ) 
   = ask B -> replace( C0, N, C ) .
  eq replace( C0 || C1, N, C )
   = replace( C0, N, C ) || replace( C1, N, C ) .
  eq replace( < N0 >[ C0 ], N, C )
   = < N0 >[ replace( C0, N, C ) ] .
  eq replace( rec( N0, C0 ), N, C ) 
   = rec( N0, C0 ) .
  eq replace( xtr( N0, C0 ), N, C )  
   = xtr( N0, replace( C0, N, C ) ) .
  eq replace( v(N0), N, C ) 
   = if (N0 ~ N) then C else v(N0) fi .
  
  op exists-store? : Cnf Aid -> Bool .
  eq exists-store?(mt, L)
   = false .
  eq exists-store?( [process, L0, C0] X, L)
   = exists-store?(X,L) .
  eq exists-store?( [store, L0, B0] X, L)
   = (L0 ~ L) or-else exists-store?(X, L) . 
endfm

--- transitions
mod SCCP is
  inc SCCP-STATE .
  pr SMT-UTIL .
  
  vars N N0 N1 : iNat .
  vars L L0 L1 : Aid .
  vars B B0 B1 : Boolean .
  vars C C0 C1 : SCCPCmd .
  vars X : Cnf .
  
  --- non-observable concurrent transitions
  eq { [ process, L0, 0 ] X }
   = { mt X } .
  eq [ store, L0, B0 ] [ store, L0, B1 ]
   = [ store, L0, B0 and B1 ] .
  
  --- observable concurrent transitions
  rl [tell] :
     { [ store, L0, B0 ] [process, L0, tell B1 ] X }
  => { [ store, L0, B0 and B1 ] [ process, L0, 0 ] X } .

 crl [ask] :
     { [ store, L0, B0 ] [ process, L0, ask B1  -> C1 ] X }
  => { [ store, L0, B0 ] [ process, L0, C1 ] X } 
  if entails(B0, B1) .

  rl [parallel] :
     { [ process, L0, C0 || C1 ] X }
  => { [ process, L0, C0 ] [ process, L0, C1 ] X } .
  
  rl [space] :
     { [ store, L0, B0 ] [ process, L0, < N0 >[ C0 ] ] X } 
  => { [ store, L0, B0 ] [ process, L0, 0 ] 
       [ store, N0 . L0, true ] [ process, N0 . L0, C0 ] X } .
  
  rl [recursion]:
     { [ process, L0, rec( N0, C0 ) ] X }
  => { [ process, L0, replace( C0, N0, rec( N0, C0 ) ) ] X } .
  
  rl [extrussion]:
     { [ process, N0 . L0, xtr( N0, C0 ) ] X }
  => { [ process, N0 . L0, 0 ] [ process, L0, C0 ] X } .
endm
\end{maude}

\section{SCCP test specification}
\label{sec.sccp.mfe}

This appendix includes the $\SCCP$ specification in Maude used to 
execute the mechanical proofs of admissibility explained in 
Section~\ref{sec.adm}, including the functional modules 
$\cde{IBOOL}$, $\cde{INAT}$, $\cde{BOOLEAN}$, $\cde{INTEGER}$,
$\cde{SMT-UTIL}$, $\cde{AGENT-ID}$, $\cde{SCCP-SYNTAX}$ and 
$\cde{SCCP-STATE}$, and the system module $\cde{SCCP}$.

\subsection{prelude-short.maude}

\begin{maude}
fmod ITRUTH-VALUE is
  sort iBool .
  op true : -> iBool [ctor] .
  op false : -> iBool [ctor] .
endfm

fmod IBOOL-OPS is
  protecting ITRUTH-VALUE .
  op _and_ : iBool iBool -> iBool [assoc comm prec 55] .
  op _or_ : iBool iBool -> iBool [assoc comm prec 59] .
  op _xor_ : iBool iBool -> iBool [assoc comm prec 57] .
  op not_ : iBool -> iBool [prec 53] .
  op _implies_ : iBool iBool -> iBool [gather (e E) prec 61] .
  vars A B C : iBool .
  eq true and A = A .
  eq false and A = false .
  eq A and A = A .
  eq false xor A = A .
  eq A xor A = false .
  eq A and (B xor C) = A and B xor A and C .
  eq not A = A xor true .
  eq A or B = A and B xor A xor B .
  eq A implies B = not(A xor A and B) .
endfm

fmod IBOOL is
  protecting IBOOL-OPS .
endfm

fmod IEXT-BOOL is
  protecting IBOOL .
  op _and-then_ : iBool iBool -> iBool [strat (1 0) gather (e E) prec 55] .
  op _or-else_ : iBool iBool -> iBool [strat (1 0) gather (e E) prec 59] .
  var B : [iBool] .
  eq true and-then B = B .
  eq false and-then B = false .
  eq true or-else B = true .
  eq false or-else B = B .
endfm

---- Natural numbers
fmod INAT is
  protecting ITRUTH-VALUE .
  sorts iZero iNzNat iNat .
  subsort iZero iNzNat < iNat .
  
  op 0 : -> iZero [ctor] .
  op s_ : iNat -> iNzNat [ctor] .
  
  --- equality enrichment
  op _~_ : iNat iNat -> iBool [comm] .
  eq 0 ~ 0 
   = true .
  eq s N:iNat ~ 0 
   = false .
  eq N:iNat ~ N:iNat
   = true .
  eq s N:iNat ~ s M:iNat
   = N:iNat ~ M:iNat .
endfm
\end{maude}

\subsection{smt.maude}

\begin{maude}
--- SMT simulation
fmod BOOLEAN is
  sort Boolean .
  op true : -> Boolean .
  op false : -> Boolean .
  
  op not_ : Boolean -> Boolean [prec 53] .
  op _and_ : Boolean Boolean -> Boolean [assoc comm gather (E e) prec 55] .
  op _xor_ : Boolean Boolean -> Boolean [assoc comm gather (E e) prec 57] .
  op _or_ : Boolean Boolean -> Boolean [assoc comm gather (E e) prec 59] .
  op _implies_ : Boolean Boolean -> Boolean [gather (e E) prec 61] .
  
  op _===_ : Boolean Boolean -> Boolean [gather (e E) prec 51] .
  op _=/==_ : Boolean Boolean -> Boolean [gather (e E) prec 51] .
  op _?_:_ : Boolean Boolean Boolean -> Boolean [gather (e e e) prec 71] .
endfm
  
fmod INTEGER is
  protecting BOOLEAN .
  sort Integer .
  op <Integers> : -> Integer .
  
  op -_ : Integer -> Integer .
  op _+_ : Integer Integer -> Integer [gather (E e) prec 33] .
  op _*_ : Integer Integer -> Integer [gather (E e) prec 31] .
  op _-_ : Integer Integer -> Integer [gather (E e) prec 33] .
  op _div_ : Integer Integer -> Integer [gather (E e) prec 31] .
  op _mod_ : Integer Integer -> Integer [gather (E e) prec 31] .
  
  op _<_ : Integer Integer -> Boolean [prec 37] .
  op _<=_ : Integer Integer -> Boolean [prec 37] .
  op _>_ : Integer Integer -> Boolean [prec 37] .
  op _>=_ : Integer Integer -> Boolean [prec 37] .
  
  op _===_ : Integer Integer -> Boolean [gather (e E) prec 51] .
  op _=/==_ : Integer Integer -> Boolean [gather (e E) prec 51] .
  op _?_:_ : Boolean Integer Integer -> Integer [gather (e e e) prec 71] .
  
  *** seems to break CVC4
  op _divisible_ : Integer Integer -> Boolean [prec 51] .
endfm
\end{maude}

\subsection{smt-util.maude}

\begin{maude}
load smt.maude
load prelude-short

fmod SMT-UTIL is
  inc INTEGER .
  pr IBOOL-OPS .
  
  op check-sat : Boolean -> iBool .
  op check-unsat : Boolean -> iBool .
  op entails : Boolean Boolean -> iBool .
  eq check-sat(B:Boolean)
   = false .
  eq check-unsat(B:Boolean)
   = not(check-sat(B:Boolean)) .
  eq entails(C1:Boolean, C2:Boolean)
   = check-unsat(C1:Boolean and not(C2:Boolean)) .
  
  --- some Boolean identities
  eq B:Boolean and true
   = B:Boolean .
  eq B:Boolean and false
   = false .
  eq B:Boolean or true
   = true .
  eq B:Boolean or false
   = B:Boolean .
  eq not((true).Boolean)
   = (false).Boolean .
  eq not((false).Boolean)
   = (true).Boolean .
endfm
\end{maude}

\subsection{agent-id.maude}

\begin{maude}
load smt-util.maude

--- agent identifier
fmod AGENT-ID is
  pr INAT .
  pr IEXT-BOOL .
  sort Aid .
  
  op root : -> Aid .
  op _._ : iNat Aid -> Aid .
  
  vars L L0 L1 L2 : Aid .
  vars N N0 N1 N2 : iNat .
  
  --- auxiliary operations
  op is-prefix? : Aid Aid -> iBool .
  eq is-prefix?(root, L)
   = true . 
  eq is-prefix?(N . L, root) 
   = false .
  eq is-prefix?(N0 . L0, N1 . L1) 
   = (N0 . L0 ~ N1 . L1) or-else is-prefix?(N0 . L0, L1) .
  
  --- equality enrichment
  op _~_ : Aid Aid -> iBool [comm] .
  eq root ~ root
   = true .
  eq root ~ N . L
   = false . 
  eq L ~ L
   = true .
  eq N . L ~ N0 . L0
   = (N ~ N0) and-then L ~ L0 .
endfm
\end{maude}

\subsection{sccp.maude}

\begin{maude}
load agent-id.maude

--- commands syntax
fmod SCCP-SYNTAX is
  pr INTEGER .
  pr AGENT-ID .
  sort SCCPCmd .
  
  op 0 : -> SCCPCmd . 
  op tell_ : Boolean -> SCCPCmd . 
  op ask_->_ : Boolean SCCPCmd -> SCCPCmd .
  op _||_ : SCCPCmd SCCPCmd -> SCCPCmd [assoc comm gather (e E) ] .
  op <_>[_] : iNat SCCPCmd -> SCCPCmd .
  op rec(_,_) : iNat SCCPCmd -> SCCPCmd .
  op xtr(_,_) : iNat SCCPCmd -> SCCPCmd .
  op v(_) : iNat -> SCCPCmd .
  
  vars C C0 C1 : SCCPCmd .
  
  --- auxiliary 
  op IF_THEN_ELSE_FI : iBool SCCPCmd SCCPCmd -> SCCPCmd [strat (1 0 0)] .
  eq IF true THEN C0 ELSE C1 FI 
   = C0 .
  eq IF false THEN C0 ELSE C1 FI 
   = C1 .
endfm

--- state syntax
fmod SCCP-STATE is
  pr SCCP-SYNTAX .
  
  sorts Cid Obj Cnf Sys .
  subsorts Obj < Cnf .
  ops store process : -> Cid .
  op [_,_,_] : Cid Aid Boolean -> Obj [ctor] .
  op [_,_,_] : Cid Aid SCCPCmd -> Obj [ctor] .
  op mt : -> Cnf [ctor] .
  op __ : Cnf Cnf -> Cnf [ctor assoc comm id: mt] .
  op {_} : Cnf -> Sys [ctor] .
  
  vars L L0 L1 : Aid .
  vars N N0 N1 : iNat .
  vars B B0 B1 : Boolean .
  vars C C0 C1 : SCCPCmd .
  vars X Y     : Cnf .
  
  --- auxiliary operations
  op replace : SCCPCmd iNat SCCPCmd -> SCCPCmd .
  eq replace( 0, N, C ) 
   = 0 . 
  eq replace( tell B, N, C ) 
   = tell B .
  eq replace( ask B -> C0, N, C ) 
   = ask B -> replace( C0, N, C ) .
  eq replace( C0 || C1, N, C )
   = replace( C0, N, C ) || replace( C1, N, C ) .
  eq replace( < N0 >[ C0 ], N, C )
   = < N0 >[ replace( C0, N, C ) ] .
  eq replace( rec( N0, C0 ), N, C ) 
   = rec( N0, C0 ) .
  eq replace( xtr( N0, C0 ), N, C )  
   = xtr( N0, replace( C0, N, C ) ) .
  eq replace( v(N0), N, C ) 
   = IF (N0 ~ N) THEN C ELSE v(N0) FI .
  
  op exists-store? : Cnf Aid -> iBool .
  eq exists-store?(mt, L)
   = false .
  eq exists-store?( [process, L0, C0] X, L)
   = exists-store?(X,L) .
  eq exists-store?( [store, L0, B0] X, L0)
   = true .
 ceq exists-store?( [store, L0, B0] X, L)
   = exists-store?(X, L)
  if L0 ~ L = false .
endfm

--- transitions
mod SCCP is
  inc SCCP-STATE .
  pr SMT-UTIL .
  
  vars N N0 N1 : iNat .
  vars L L0 L1 : Aid .
  vars B B0 B1 : Boolean .
  vars C C0 C1 : SCCPCmd .
  vars X : Cnf .
  
  --- non-observable concurrent transitions
  eq { [ process, L0, 0 ] X }
   = { mt X } .
  eq [ store, L0, B0 ] [ store, L0, B1 ]
   = [ store, L0, B0 and B1 ] .

  --- observable concurrent transitions
  rl [tell] :
     { [ store, L0, B0 ] [process, L0, tell B1 ] X }
  => { [ store, L0, B0 and B1 ] [ process, L0, 0 ] X } .

 crl [ask] :
     { [ store, L0, B0 ] [ process, L0, ask B1  -> C1 ] X }
  => { [ store, L0, B0 ] [ process, L0, C1 ] X } 
  if entails(B0, B1) = true .

  rl [parallel] :
     { [ process, L0, C0 || C1 ] X }
  => { [ process, L0, C0 ] [ process, L0, C1 ] X } .

  rl [space] :
     { [ store, L0, B0 ] [ process, L0, < N0 >[ C0 ] ] X } 
  => { [ store, L0, B0 ] [ process, L0, 0 ] 
  	   [ store, N0 . L0, true ] [ process, N0 . L0, C0 ]  X } .

  rl [recursion]:
     { [ process, L0, rec( N0, C0 ) ] X }
  => { [ process, L0, replace( C0, N0, rec( N0, C0 ) ) ] X } .

  rl [extrusion]:
     { [ process, N0 . L0, xtr( N0, C0 ) ] X }
  => { [ process, N0 . L0, 0 ] [ process, L0, C0 ] X } .
endm
\end{maude}

\end{document}